\documentclass[12pt,twoside]{amsart}
\usepackage{amsmath,amsthm,amsmath,amsrefs}

\usepackage{calc}
\usepackage{setspace}
\usepackage{amsbsy,amsmath,amsfonts,amsthm,amssymb,color}
\usepackage[margin=1in]{geometry}
\newtheorem{innercustomtheorem}{Theorem}
\newenvironment{customtheorem}[1]
  {\renewcommand\theinnercustomtheorem{#1}\innercustomtheorem}
  {\endinnercustomtheorem}

\numberwithin{equation}{section}
\theoremstyle{plain}
\newtheorem{theorem}{Theorem}[section]
\newtheorem{lemma}[theorem]{Lemma}

\newtheorem{definition}[theorem]{Definition}

\newtheorem{remark}[theorem]{Remark}
\newtheorem{problem}[theorem]{Problem}
\newtheorem{proposition}[theorem]{Proposition}

\theoremstyle{definition}

\usepackage{mathrsfs}

\newcommand{\ee}{\varepsilon}

\newcommand{\bN}{\mathbb{N}}

\newcommand{\cF}{\mathcal{F}}

\newcommand{\vertiii}[1]{{\left\vert\kern-0.25ex\left\vert\kern-0.25ex\left\vert #1 
    \right\vert\kern-0.25ex\right\vert\kern-0.25ex\right\vert}}

\newcommand{\cB}{\mathcal{B}}

\newcommand{\la}{\langle}
\newcommand{\ra}{\rangle}

\newcommand{\cH}{\mathcal{H}}
\newcommand{\cA}{\mathcal{A}}

\newcommand{\what}{\widehat}
\newcommand{\bofh}{\cB(\cH)}

\newcommand{\cM}{\mathcal{M}}
\newcommand{\cU}{\mathcal{U}}

\newcommand{\cG}{\mathcal{G}}
\newcommand{\cR}{\mathcal{R}}

\newcommand{\cE}{\mathcal{E}}

\newcommand{\tr}{\text{tr}}

\usepackage{tikz}
\usepackage{tikz-cd}

\title[Approximate quantum 3-colorings of graphs]{Approximate quantum 3-colorings of graphs and the quantum Max 3-Cut problem}

\author{Samuel J. Harris}

\address{Northern Arizona University\\
Department of Mathematics \& Statistics\\
801 S. Osborne Dr.\\
Flagstaff, AZ\\
86011 USA}
\email{samuel.harris@nau.edu}

\begin{document}

\begin{abstract}
We prove that, to each synchronous non-local game $\mathcal{G}=(I,O,\lambda)$ with $|I|=n$ and $|O|=m \geq 3$, there is an associated graph $G_{\lambda}$ for which approximate winning strategies for the game $\mathcal{G}$ and the $3$-coloring game for $G_{\lambda}$ are preserved. That is, using a similar graph to previous work of the author (Ann. Henri Poincar\'{e}, 2024), any synchronous strategy for $\text{Hom}(G_{\lambda},K_3)$ that wins the game with probability $1-\varepsilon$ with respect to the uniform probability distribution on the edges, yields a strategy in the same model that wins the game $\cG$ with respect to the uniform distribution with probability at least $1-h(n,m)\varepsilon^{\frac{1}{2}}$, where $h$ is a polynomial in $n$ and $2^m$. As an application, we prove that the gapped promise problem for quantum $3$-coloring is undecidable. Moreover, we prove that there exists an $\alpha \in (0,1)$ for which determining whether the non-commutative Max-$3$-Cut of a graph is $|E|$ or less than $\alpha |E|$ is RE-hard, thus giving a positive answer to a problem posed by Culf, Mousavi and Spirig (arXiv:2312.16765), along with evidence for a sharp computability gap in the non-commutative Max-$3$-Cut problem. We also prove that there is some $\alpha \in (0,1)$ such that determining the non-commutative (respectively, commuting operator framework) versions of the Max-$3$-Cut of a graph within a factor of $\alpha$ is uncomputable. All of these results avoid use of the unique games conjecture.
\end{abstract}

\maketitle

\section{Introduction}

The Max Cut problem is a famous problem in graph theory: given an undirected graph, find a partition of the vertices into two subsets in such a way that maximizes the number of edges cut. It is well-known that computing the Max Cut of a graph is NP-hard. Much work has gone into seeing how well one can \textit{approximate} the value of Max Cut in polynomial time. The work of Goemans and Williamson \cite{GW95} shows that there is a polynomial time algorithm for approximating $\text{Max-Cut}(G)$ for a graph $G$ within a ratio of $0.878$. If one assumes the validity of the Unique Games Conjecture, then there is no polynomial time algorithm for approximating $\text{Max-Cut}(G)$ within a ratio larger than the Goemans-Williamson constant, unless $P=NP$ \cite{KKMO07}. 

One way of phrasing the Max Cut problem for a graph $G$ is by the following (commutative) polynomial optimization problem:
\[ \text{Max-Cut}(G)=\max_{x_j \in \{-1,1\}} \sum_{(j,k) \in E(G)} \frac{1-x_ix_j}{2},\]
where $E(G)$ denotes the set of (undirected) edges for $G$. Here, we interpret this sum as summing over each edge once. The terms involved do not depend on the ordering $(i,j)$ or $(j,i)$ for an edge. The non-commutative Max Cut (or NC-Max-Cut) is the relaxation of the above problem to the setting where the $X_j$'s are self-adjoint unitaries (in particular, they have eigenvalues $\pm 1$) in a matrix algebra:
\[ \text{NC-Max-Cut}(G)=\sup \sum_{(j,k) \in E(G)} \frac{1-\tr_d(X_iX_j)}{2},\]
where the supremum is taken over all $d \in \mathbb{N}$ and over all choices of self-adjoint unitaries $X_1,...,X_n \in M_d$, and where $\tr_d$ denotes the normalized trace on $M_d$ (i.e., satisfying $\tr_d(I_d)=1$). Unlike the Max-Cut problem, Tsirelson's work \cite{Ts87,Ts93} shows that there is a polynomial time algorithm for computing the NC-Max-Cut of a graph.

Related to this problem is the Max-3-Cut problem, where now one tries to maximize the number of edges cut when partitioning the vertices into 3 subsets. Again, there is an equivalent statement of this problem in terms of (commutative) polynomial optimization:
\[ \text{Max-3-Cut}(G)=\max_{x_j \in \{1,\omega,\omega^2\}} \sum_{(j,k) \in E(G)} \frac{2-x_ix_j^*-x_i^*x_j}{3},\]
where $\omega=\exp\left(\frac{2\pi i}{3}\right)$ is a primitive third root of unity. The problem of computing Max-3-Cut is NP-hard. On the other hand, Max-3-Cut can be approximated in polynomial time up to a constant ratio of 0.836 due to several works \cite{FJ95,GW04,KPW04}. By Khot et. al \cite{KKMO07}, if the unique games conjecture is true, then this is the best possible constant.

As with the Max Cut, there is a natural non-commutative version of the Max-3-Cut, given by
\[ \text{NC-Max-3-Cut}(G)=\sup \sum_{(j,k) \in E(G)} \frac{2-\tr_d(X_iX_j^*)-\tr_d(X_i^*X_j)}{3},\]
where the supremum is over all $d \in \mathbb{N}$ and over all choices of matrices $X_j \in M_d$ for $j \in V(G)$, satisfying $X_j^*X_j=I_d=X_j^3$ (that is, order $3$ unitaries). For the non-commutative Max-3-Cut, a polynomial time algorithm for approximating the value up to a ratio of $0.864$ is known \cite{CMS23}. Determining whether $\text{NC-Max-3-Cut}(G)$ is equal to $|E|$ or is less than $|E|$ is known to be undecidable \cite{Ji13,Ha24}. On the other hand, is there an upper bound on how well can one efficiently approximate the non-commutative max 3 cut?

Closely related to this problem is that of approximately winning correlations for the synchronous non-local game for $3$-coloring a graph. A two-player one round non-local game $\cG=(I,O,\lambda)$ is one where two non-communicating players are asked questions from the finite set $I$ by a verifier and required to give correct answers from the finite set $O$, based upon a rule function $\lambda:O \times O \times I \times I \to \{0,1\}$. The players win the round if $\lambda(a,b,x,y)=1$, and lose if $\lambda(a,b,x,y)=0$. The players are allowed to use entanglement and measurements to potentially improve their chances of success. There are several natural models for considering the probability distribution $p(a,b|x,y)$ that the players answer $a$ and $b$, respectively, given that they were asked $x$ and $y$, respectively. In general, this probability will satisfy the following:
\begin{itemize}
\item $p(a,b|x,y) \geq 0$ for all $a,b,x,y$ and $\displaystyle \sum_{a,b \in O} p(a,b|x,y)=1$ for all $a,b \in O$;
\item $\displaystyle \sum_{b \in O} p(a,b|x,y)$ is independent of $y$;
\item $\displaystyle \sum_{a \in O} p(a,b|x,y)$ is independent of $x$.
\end{itemize}
Such a probability distribution $p=(p(a,b|x,y))$ (called a \textit{correlation}) is called \textbf{winning} or \textbf{perfect} for $\cG$ if $p(a,b|x,y)=0$ whenever $\lambda(a,b,x,y)=0$. In other words, for perfect correlations, the probability of winning the game with the distribution $p$ and a probability distribution $\pi(x,y)$ on the question pairs, given by
\[ \omega(\cG,\pi,p)=\sum_{x,y \in I} \pi(x,y) \sum_{a,b \in O} \lambda(a,b,x,y)p(a,b|x,y),\]
is $1$. This paper is focused on \textbf{synchronous games}, which are games where $\lambda(a,b,x,x)=0$ for all $x \in I$ and $a \neq b$. Much is known about these games, and winning correlations for these games must be \textbf{synchronous correlations} (i.e. $p(a,b|x,x)=0$ for all $a \neq b$ and $x \in I$). In general, there are four usual models for synchronous correlations: local, quantum, quantum approximate and quantum commuting (the sets of such correlations with $n$ questions and $m$ answers are denoted by $C_{loc}^s(n,m)$, $C_q^s(n,m)$, $C_{qa}^s(n,m)$ and $C_{qc}^s(n,m)$, respectively). For all of these sets, a typical element $p$ looks like $(p(a,b|x,y))=(\tau(E_{a,x}E_{b,y}))$, where $\tau$ is a (faithful, normal) tracial state on a von Neumann algebra $\cM$, and $\{E_{a,x}\}_{a \in O}$ is a projection-valued measure (PVM) in $\cM$ for each $x \in I$, with extra arrangements that can be made on what $\cM$ is based on $t \in \{loc,q,qa,qc\}$ \cite{PSSTW16,KPS18}. For a synchronous game $\cG$ with a prior probability distribution $\pi(x,y)$ on the question pairs, and for an element $p \in C_t^s(n,m)$ ($t \in \{loc,q,qa,qc\}$), we define $\omega_t^s(\cG,\pi,p)=\sum_{x,y \in I} \pi(x,y) \sum_{a,b \in O} \lambda(a,b,x,y)p(a,b|x,y)$, which is the probability that the players win the game $\cG$ using the correlation $p$, assuming the verifier has the prior distribution $\pi$. The synchronous $t$-value of $\cG$ is simply defined as
\[ \omega_t^s(\cG,\pi)=\sup \{ \omega_t^s(\cG,\pi,p): p \in C_t^s(n,m)\}.\]

The relevant non-local game for this paper is the $3$-coloring game of a simple undirected graph $G=(V,E)$ with no loops, often denoted $\text{Hom}(G,K_3)$. For this game, the question set is $I=V$ and the answer set is $O=\{1,2,3\}$, and the rule function is given by
\[ \lambda(a,b,x,y)=\begin{cases} 0 & x=y \text{ and } a \neq b \\ 0 & x \sim y \text{ and } a=b \\
1 & \text{otherwise},\end{cases}\]
where $x\sim y$ means that $(x,y) \in E$. For a graph $G$, one has $\text{Max-}3\text{-Cut}(G) \leq |E|$, with equality if and only if there is a perfect loc strategy for the $3$-coloring game for $G$. (For the NC-Max-3-Cut, equality holds if and only if there is a perfect $qa$ strategy for the $3$-coloring game for $G$.) For convenience, we will always work with the prior distribution $\pi_{edges}$ that is uniform on vertex pairs that form edges in $G$; i.e., $\pi_{edges}(x,y)=\frac{1}{2|E|}$ if $(x,y) \in E$ and $\pi_{edges}(x,y)=0$ otherwise.

In this paper, we show that for each synchronous non-local game $\cG=(I,O,\lambda)$ with $|I|=n$ and $|O|=m \geq 3$, there is an associated graph $G_{\lambda}$ so that approximately winning synchronous strategies are preserved between the two games. In other words, the following holds:
\begin{customtheorem}{\ref{theorem: value of coloring to value of game}}
Let $\ee>0$. If $t \in \{loc,q,qa,qc\}$ and if $\omega_t^s(\text{Hom}(G_{\lambda},K_3),\pi_{edges}) \geq 1-\ee$, then $\omega_t^s(\cG,\pi_u) \geq 1-\text{poly}(n,2^m)\ee^{\frac{1}{2}}$, where $\pi_u$ is the uniform distribution on question pairs for $\cG$.
\end{customtheorem}

The converse also holds: for each $\ee>0$ and $t \in \{loc,q,qa,qc\}$, if $\omega_t^s(\cG,\pi_u) \geq 1-\ee$, then $\omega_t^s(\text{Hom}(G_{\lambda},K_3),\pi_{edges}) \geq 1-\text{poly}(n,m)\ee$. (See Theorem \ref{theorem: forward direction}.)

We use Theorem \ref{theorem: value of coloring to value of game} to prove the undecidability of the so-called gapped $(1,\alpha)$-promise problem for the synchronous $t$-value of the $3$-coloring game. This promise problem is the problem of deciding whether the synchronous $t$ value of the $3$-coloring game $\text{Hom}(G,K_3)$ is equal to $1$, or is at most $\alpha$, if promised that one of these two outcomes occur. We prove that this gapped promise problem is undecidable for some $\alpha$ (in particular, for some $\alpha$ the gapped $(1,\alpha)$-promise problem for the synchronous $q$-value of the $3$-coloring game is RE-hard).

As a result, one cannot efficiently compute the NC-Max-3-Cut up to aribtrary precision: there exists an $\alpha \in (0,1)$ for which approximating $\text{NC-Max-3-Cut}(G)$ up to a constant factor of $\alpha$ is uncomputable, in the following sense:
\begin{customtheorem}{\ref{theorem: sharp computability gap}}
There is an $\alpha \in (0,1)$ such that the decision problem of deciding whether $\text{NC-Max-3-Cut}(G)$ is equal to $|E|$ or is less than $\alpha|E|$, when promised that one of these two possibilities occurs, is undecidable.
\end{customtheorem}

Our method for proving the above theorem is by linking this problem to the problem of deciding whether a (synchronous) non-local game has quantum value $1$ or quantum value less than $1-\ee$, when promised that one of those two outcomes occurs. The latter result, by a breakthrough work of Ji et. al \cite{JNVWY20}, is RE-complete, and so the former result is RE-hard. Our main tool for the reduction uses previous work of the author \cite{Ha24}, which establishes a $*$-equivalence of a synchronous game $\cG$ and an associated $3$-coloring game for an associated graph $G_{\lambda}$. Given a synchronous non-local game $\cG=(I,O,\lambda)$, one can construct its game algebra $\cA(\cG)$, which encodes all winning strategies for the game. The algebra $\cA(\cG)$ is the universal unital $*$-algebra generated by elements $e_{a,x}$, $a \in O$, $x \in I$, subject to the properties that
\begin{itemize}
\item $e_{a,x}^2=e_{a,x}=e_{a,x}^*$ for all $a,x$;
\item $\displaystyle \sum_{a \in O} e_{a,x}=1$ for all $x \in I$; and
\item $\lambda(a,b,x,y)=0 \implies e_{a,x}e_{b,y}=0$.
\end{itemize}
Two games $\cG_1$ and $\cG_2$ are $*$-equivalent if there are unital $*$-homomorphisms $\cA(\cG_1) \to \cA(\cG_2)$ and $\cA(\cG_2) \to \cA(\cG_1)$. While a seemingly weak condition (as the linking homomorphisms need not even be injective), when two synchronous games are $*$-equivalent, the existence of a winning strategy for one of the games in a typical model (i.e. loc,q,qa,qc) is equivalent to the existence of such a winning strategy for the other game, in the same model. We refer the reader to works such as \cite{Go21,HMPS19,Ha22,Ha24,KPS18} for more information and examples of $*$-equivalence, and \cite{KPS18,PSSTW16} for background information on synchronous correlations.
 
The graph $G_{\lambda}$ used in \cite{Ha24} has a control triangle $\Delta=\{A,B,C\}$ that forces many of the relations in the game algebra. Namely, there is a unital $*$-homomorphism $\pi:\cA(\text{Hom}(G_{\lambda},K_3)) \to \cA(\cG)$. Conversely, for each choice of $\{i,j,k\}=\{1,2,3\}$, there is a unital $*$-homomorphism $\rho_{ijk}:\cA(\cG) \to e_{i,A}e_{j,B}e_{k,C}\cA(\text{Hom}(G_{\lambda},K_3))e_{i,A}e_{j,B}e_{k,C}$. This theorem uses the fact that each of the projections $e_{i,A}$, $e_{j,B}$ and $e_{k,C}$ are in the center of $\cA(\text{Hom}(G_{\lambda},K_3))$; indeed, ensuring this property is perhaps one of the more difficult parts of constructing the graph $G_{\lambda}$. As one also can show that the sum of the projections $e_{i,A}e_{j,B}e_{k,C}$ over $\{i,j,k\}=\{1,2,3\}$ is $1$, adding the $\rho_{ijk}$'s shows that there is a unital $*$-homomorphism $\rho:\cA(\cG) \to \cA(\text{Hom}(G_{\lambda},K_3))$. In particular, $\cG$ and $\text{Hom}(G_{\lambda},K_3)$ are $*$-equivalent.

We show that a slight modification of this mapping (which is unchanged in the context of perfect strategies for the games) preserves approximately winning synchronous strategies, when using the uniform distribution on question pairs for $\cG$ and the uniform distribution on vertex pairs that form edges for $\text{Hom}(G_{\lambda},K_3)$ (see Theorems \ref{theorem: forward direction} and \ref{theorem: value of coloring to value of game}). We note here that our use of the edge distribution is mainly for convenience and for its relevance to the Max-3-Cut problem, and one could arrive at a similar result (although with a more complicated estimate) for any distribution that is positive for all edges, assuming the strategy is synchronous. These equivalences hold for each of the models loc,q,qa,qc.

We note that our work gives an alternate proof of the existence of an $\alpha$ so that the gapped promise problem of determining whether $\text{Max-3-Cut}(G)$ is equal to $|E|$ or less than $\alpha |E|$, when promised that one of these occurs, cannot be solved in polynomial time, assuming $P \neq NP$.

The rest of the paper is organized as follows. In section \ref{section: approximations}, we collect preliminary facts regarding approximations in the $2$-norm with respect to a faithful, normal tracial state $\tau$ on a von Neumann algebra $\mathcal{M}$. In section \ref{section: graph definition} we give a slightly modified definition of the graph $G_{\lambda}$ arising from earlier work of the author \cite{Ha24}, as we will use its structure extensively. (The modifications are only to simplify some of the arguments.) In section \ref{section: from G to Glambda} we show how to go from approximately winning synchronous strategies for $\cG$ to approximately winning synchronous strategies for $\text{Hom}(G_{\lambda},K_3)$. Section \ref{section: from Glambda to G} deals with the problem of going from approximately winning synchronous strategies for $\text{Hom}(G_{\lambda},K_3)$ to approximately winning synchronous strategies for $\cG$. Lastly, in Section \ref{section: max 3 cut} we apply our results and estimates to the different versions of the Max-3-Cut problem.

\section{Approximation preliminaries}
\label{section: approximations}

In this section, we gather many of the approximation results that we will need in showing that approximately winning strategies are preserved in the game equivalence between a synchronous non-local game $\cG$ and the $3$-coloring game for its associated graph $G_{\lambda}$ that is defined in Section \ref{section: graph definition}. For the remainder of this section, we fix a finite von Neumann algebra $\mathcal{M}$ equipped with a faithful normal tracial state $\tau$. We will use the $2$-norm on $\mathcal{M}$ given by $\|A\|_2=\tau(A^*A)^{\frac{1}{2}}$, while $\|A\|$ will denote the operator norm of $A$. There are a few basic facts that we will need:
\begin{itemize}
\item If $A,B,C \in \mathcal{M}$, then $|\tau(ABC)| \leq \|ABC\|_2 \leq \|A\|\|B\|_2\|C\|$.
\item If $A,B \in \cM$ and $A^*A \leq B^*B$, then $\|A\|_2 =\tau(A^*A)^{\frac{1}{2}} \leq \tau(B^*B)^{\frac{1}{2}}=\|B\|_2$.
\end{itemize}
We will also use the notation $A \approx_{\ee} B$ to mean that $\|A-B\|_2<\ee$. Note that, if this is the case and $C$ is another operator in $\mathcal{M}$, then
\[ CA \approx_{\|C\|\ee} CB \text{ and } AC \approx_{\|C\|\ee} BC.\] 

Most of the results in this section can be interpreted as facts about synchronous strategies that approximately win the $3$-coloring game for triangles and for triangular prisms. The exact versions of these results can be found in \cite{Ji13,Ha24}.

The first result shows that three self-adjoint operators that are approximately idempotent and sum to zero are approximately zero.

\begin{proposition}
\label{proposition: three almost-projections summing to zero}
Suppose that $A_1,A_2,A_3 \in \mathcal{M}$ are self-adjoint. If $A_1+A_2+A_3=0$, then for each $i=1,2,3$ we have
\[\|A_i\|_2 \leq \sqrt{3}\left(\sum_{j=1}^3 \|A_j^2-A_j\|_2\right)^{\frac{1}{2}}.\]
\end{proposition}

\begin{proof}
We rearrange to get $A_1+A_2=-A_3$. Squaring both sides gives
\[ A_1^2+A_1A_2+A_2A_1+A_2^2=A_3^2.\]
Applying the trace gives
\[ 2\tau(A_1A_2)+\tau(A_1^2)+\tau(A_2^2)=\tau(A_3^2).\]
Since $\tau(A_1)+\tau(A_2)+\tau(A_3)=0$, we then have
\begin{align*}
|2\tau(A_1A_2)-2\tau(A_3)|&=|2\tau(A_1A_2)+\tau(A_1)+\tau(A_2)-\tau(A_3)| \\
&\leq |2\tau(A_1A_2)+\tau(A_1^2)+\tau(A_2^2)-\tau(A_3^2)| \\
&\,\,\,\,+|\tau(A_1^2-A_1)|+|\tau(A_2^2-A_2)|+|\tau(A_3^2-A_3)| \\
&\leq \sum_{i=1}^3 \|A_i^2-A_i\|_2.
\end{align*}
Using a symmetric argument for each choice of $i,j,k$ with $\{i,j,k\}=\{1,2,3\}$ and the fact that $A_1+A_2+A_3=0$, it follows that
\begin{align*}
\left| \sum_{\substack{1 \leq i,j \leq 3 \\ i \neq j}} \tau(A_iA_j)\right|&=\left| \sum_{1 \leq i<j \leq 3} 2\tau(A_iA_j)-\sum_{i=1}^3 2\tau(A_i)\right| \\
&\leq \sum_{\{i,j,k\}=\{1,2,3\}} |2\tau(A_iA_j)-2\tau(A_k)| \\
&\leq 3 \sum_{i=1}^3 \|A_i^2-A_i\|_2.
\end{align*}
Finally,
\[ 0=\|A_1+A_2+A_3\|_2^2=\|A_1\|_2^2+\|A_2\|_2^2+\|A_3\|_2^2+\sum_{\substack{1 \leq i,j \leq 3 \\ i \neq j}} \tau(A_iA_j),\]
so that $(\|A_1\|_2^2+\|A_2\|_2^2+\|A_3\|_2^2)^{\frac{1}{2}} \leq \sqrt{3}\left(\sum_{i=1}^3 \|A_i^2-A_i\|_2\right)^{\frac{1}{2}}$. It follows that, for each $i \in \{1,2,3\}$, we have
\[ \|A_i\|_2 \leq (\|A_1\|_2^2+\|A_2\|_2^2+\|A_3\|_2^2)^{\frac{1}{2}}\leq \sqrt{3} \left(\sum_{j=1}^3 \|A_j^2-A_j\|_2\right)^{\frac{1}{2}},\]
as desired.
\end{proof}

The next proposition shows that, given two approximate PVMs $\{A_1,A_2,A_3\}$ and $\{B_1,B_2,B_3\}$, if $A_i$ and $B_i$ approximately commute for each $i=1,2,3$, then $A_i$ and $B_j$ approximately commute for any $i,j$. This is analogous to the fact that, for perfect operator strategies for the $3$-coloring game, projections corresponding to adjacent vertices commute.

\begin{proposition}
\label{proposition: commutators of projections}
Let $A_1,A_2,A_3$ and $B_1,B_2,B_3$ be projections in $\mathcal{M}$. Then for each $i \neq j$ in $\{1,2,3\}$, we have
\[ \|A_iB_j-B_jA_i\|_2 \leq 2\left(\left\| 1-\sum_{i=1}^3 A_i\right\|_2 + \left\| 1-\sum_{i=1}^3 B_i\right\|_2 + \sum_{i=1}^3 \|A_iB_i-B_iA_i\|_2 \right).\]
In particular,
\[ \sum_{\substack{1 \leq i,j \leq 3\\ i \neq j}} \|A_iB_j-B_jA_i\|_2 \leq 12\left(\left\| 1-\sum_{i=1}^3 A_i\right\|_2+\left\|1-\sum_{i=1}^3 B_i\right\|_2+\sum_{i=1}^3 \|A_iB_i-B_iA_i\|_2\right).\]
\end{proposition}

\begin{proof}
For convenience, define $\eta_1=\left\| 1-\sum_{i=1}^3 A_i\right\|_2$ and $\eta_2=\left\| 1-\sum_{i=1}^3 B_i \right\|_2$. By symmetry, to compute $\|A_iB_j-B_jA_i\|_2$ for $i \neq j$, we may as well assume that $i=1$ and $j=2$. We have
\begin{align*}
A_1B_2&\approx_{\eta_1} A_1B_2(A_1+A_2+A_3) \\
&=A_1B_2A_1+A_1B_2A_2+A_1B_2A_3 \\
&\approx_{\eta_2} A_1B_2A_1+A_1B_2A_2+A_1(1-B_1-B_3)A_3 \\
&=A_1B_2A_1+A_1B_2A_2-A_1B_1A_3-A_1B_3A_3,
\end{align*}
where the last line follows since $A_iA_j=0$ for $i \neq j$. Thus,
\begin{align*}
\|A_1B_2-A_1B_2A_1\|_2&\leq \eta_1+\eta_2+\|A_1B_2A_2\|_2+\|A_1B_1A_3\|_2+\|A_1B_3A_3\|_2 \\
&\leq \eta_1+\eta_2+\sum_{i=1}^3 \|A_iB_i-B_iA_i\|_2,
\end{align*}
with the last line following again by the fact that $A_iA_j=0$ for $i \neq j$. Applying the triangle inequality again gives
\begin{align*}
\|A_1B_2-B_2A_1\|_2 &\leq \|A_1B_2-A_1B_2A_1\|_2+\|A_1B_2A_1-B_2A_1\|_2 \\
&=2\|A_1B_2-A_1B_2A_1\|_2 \\
&\leq 2\left(\left\| 1-\sum_{i=1}^3 A_i\right\|_2+\left\| 1-\sum_{i=1}^3 B_i\right\|_2+\sum_{i=1}^3 \|A_iB_i-B_iA_i\|_2\right).
\end{align*}
The final statement of the proposition follows immediately.
\end{proof}

In perfect operator strategies for $3$-coloring a triangle, for each fixed color, the sum of projections over the vertices in the triangle is $1$. In other words, the projections for a $3$-coloring of a triangle form a $3 \times 3$ quantum permutation. The next proposition gives an approximate form of this fact.

\begin{proposition}
\label{proposition: approx quantum permutation}
Suppose $\{P_{a,x}: a=1,2,3\}$ is a PVM in $\mathcal{M}$ for each $x=1,2,3$. Then for each $a=1,2,3$,
\[\left\| \sum_{x=1}^3 P_{a,x}-1\right\|_2 \leq \sqrt{3}\left(\sum_{b=1}^3 \sum_{\substack{1 \leq x,y \leq 3 \\ x \neq y}} \|P_{b,x}P_{b,y}\|_2\right)^{\frac{1}{2}}.\]
In particular, if $\|P_{a,x}P_{a,y}\|_2=\|P_{b,x}P_{b,y}\|_2$ for each $x \neq y$ and $a \neq b$, then
\[ \left\| \sum_{x=1}^3 P_{a,x}-1\right\|_2 \leq 3\left(\sum_{\substack{1 \leq x,y \leq 3 \\ x \neq y}} \|P_{b,x}P_{b,y}\|_2\right)^{\frac{1}{2}}, \, \forall a,b=1,2,3.\]
\end{proposition}

\begin{proof}
Define $Q_a=1-\sum_{x=1}^3 P_{a,x}$ for $a=1,2,3$. Then $Q_a$ is self-adjoint. Using the fact that each $P_{a,x}$ is a projection, we have
\[\|(1-Q_a)^2-(1-Q_a)\|_2=\left\| \sum_{x,y=1}^3 P_{a,x}P_{a,y}-\sum_{x=1}^3 P_{a,x} \right\|\leq \sum_{\substack{1 \leq x,y \leq 3 \\ x \neq y}} \|P_{a,x}P_{a,y}\|_2.\]
It follows that, for each $a=1,2,3$,
\[\|Q_a^2-Q_a\|_2=\|(1-Q_a)^2-(1-Q_a)\|_2 \leq \sum_{\substack{1 \leq x,y \leq 3 \\ x \neq y}} \|P_{a,x}P_{a,y}\|_2.\] 
Next, we note that
\[ \sum_{a=1}^3 Q_a=3-\sum_{a=1}^3 \sum_{x=1}^3 P_{a,x}=3-\sum_{x=1}^3 1=0.\]
By Proposition \ref{proposition: three almost-projections summing to zero}, it follows that, for each $a=1,2,3$, 
\[\|Q_a\|_2 \leq \sqrt{3} \left(\sum_{b=1}^3 \|Q_b^2-Q_b\|_2\right)^{\frac{1}{2}} \leq \sqrt{3}\left(\sum_{b=1}^3 \sum_{\substack{1 \leq x,y \leq 3 \\ x \neq y}} \|P_{b,x}P_{b,y}\|_2\right)^{\frac{1}{2}}.\]
The last claim easily follows.
\end{proof}

The next result shows that approximately winning strategies for the $3$-coloring game for a triangular prism force projections for non-adjacent vertices to approximately commute (for the exact version, see \cite{Ji13,Ha24}).

\begin{proposition}
\label{proposition: almost commuting almost quantum permutations}
Let $\{P_{i,j}\}_{i,j=1}^3$ and $\{Q_{i,j}\}_{i,j=1}^3$ be projections in $\mathcal{M}$ such that $\sum_{i=1}^3 P_{i,j}=\sum_{i=1}^3 Q_{i,j}=1$ for all $j$.
\begin{enumerate}
\item For each $i=1,2,3$ and $j \neq \ell$,
\[\|[P_{ij},Q_{i\ell}]\|_2 \leq 2\left(\left\| 1-\sum_{x=1}^3 P_{i,x}\right\|_2+\left\|1-\sum_{x=1}^3 Q_{i,x}\right\|_2+\sum_{x=1}^3 2\|P_{i,x}Q_{i,x}\|_2\right).\]
\item For each $j=1,2,3$ and $i \neq k$,
\[ \|[P_{ij},Q_{kj}]\|_2 \leq 4\sum_{a=1}^3 \|P_{a,j}Q_{a,j}\|_2.\]
\item If $i \neq k$ and $j \neq \ell$, then
\[ \|[P_{ij},Q_{k\ell}]\|_2 \leq 4\sum_{a=1}^3 \left(\left\|1-\sum_{x=1}^3 P_{a,x}\right\|_2+\left\|1-\sum_{x=1}^3 Q_{a,x}\right\|_2+\sum_{x=1}^3 2\|P_{a,x}Q_{a,x}\|_2\right).\]
\end{enumerate}
\end{proposition}

\begin{proof}
If $j \neq \ell$, then using Proposition \ref{proposition: commutators of projections} with the projections $\{P_{ij}: 1 \leq j \leq 3\}$ and $\{Q_{i\ell}:1 \leq \ell \leq 3\}$, we obtain
\[ \|[P_{ij},Q_{i\ell}]\|_2 \leq 2\left(\left\| 1-\sum_{x=1}^3 P_{i,x}\right\|_2+\left\|1-\sum_{x=1}^3 Q_{i,x}\right\|_2+\sum_{x=1}^3 2\|P_{i,x}Q_{i,x}\|_2\right),\]
which is the first case. In the case when $i \neq k$ but $j=\ell$, a similar argument with the PVM's $\{ P_{ij}: 1 \leq i \leq 3\}$ and $\{Q_{ij}: 1 \leq i \leq 3\}$ yields, for $i \neq k$,
\[ \|[P_{ij},Q_{kj}]\|_2 \leq 4\sum_{a=1}^3 \|P_{aj}Q_{aj}\|_2.\]
In the last case, we have $i \neq k$ and $j \neq \ell$. By Proposition \ref{proposition: commutators of projections} with the PVMs $\{P_{aj}: 1 \leq i \leq 3\}$ and $\{Q_{a\ell}: 1 \leq x \leq 3\}$, one has
\[\|[P_{ij},Q_{k\ell}]\|_2 \leq 2\sum_{a=1}^3 \|[P_{aj},Q_{a\ell}]\|_2.\]
Since $j \neq \ell$, by the first case we obtain
\[ \|[P_{ij},Q_{k\ell}]\|_2 \leq 4\sum_{a=1}^3 \left( \left\| 1-\sum_{x=1}^3 P_{a,x}\right\|_2+\left\| 1-\sum_{x=1}^3 Q_{a,x} \right\|_2+\sum_{x=1}^3 2\|P_{a,x}Q_{a,x}\|_2 \right),\]
which is the last case of the desired inequality.
\end{proof}

In a perfect strategy for $3$-coloring a triangle $\{A,B,C\}$, one can show that all projections corresponding to these vertices commute with each other. Moreover, each operator of the form $P_{i,A}P_{j,B}P_{k,C}$, for $\{i,j,k\}=\{1,2,3\}$, is a projection and their sum is the identity. In the approximate setting, the projections corresponding to $\{A,B,C\}$ need not commute. Hence, the operators $P_{i,A}P_{j,B}P_{k,C}$ need not be projections, and may not even be positive operators. However, a suitable modification (which amounts to the same operators for perfect strategies) yields positive operators that nearly sum to the identity.

\begin{proposition}
\label{proposition: sum of cut down projections}
Let $\{P_{i,A}\}_{i=1}^3$, $\{P_{i,B}\}_{i=1}^3$ and $\{P_{i,C}\}_{i=1}^3$ be PVMs in $\mathcal{M}$. Then
\[\left\|1-\displaystyle \sum_{\{i,j,k\}=\{1,2,3\}} P_{i,A}P_{j,B}P_{k,C}P_{j,B}P_{i,A}\right\|_2 \leq 159 \sum_{i=1}^3 (\|P_{i,A}P_{i,B}\|_2+\|P_{i,A}P_{i,C}\|_2+\|P_{i,B}P_{i,C}\|_2)\]
\end{proposition}

\begin{proof}
Note that, since each of $\{P_{i,A}\}_{i=1}^3$, $\{P_{i,B}\}_{i=1}^3$ and $\{P_{i,C}\}_{i=1}^3$ is a PVM, we have
\[ \sum_{i,j,k=1}^3 P_{i,A}P_{j,B}P_{k,C}=1,\]
so that, multiplying this expression with its adjoint gives
\[ 1=\sum_{i,j,k=1}^3 \sum_{\alpha,\beta,\gamma=1}^3 P_{i,A}P_{j,B}P_{k,C}P_{\gamma,C}P_{\beta,B}P_{\alpha,A}.\]
If $k \neq \gamma$, then $P_{k,C}P_{\gamma,C}=0$, so this sum becomes
\[ 1=\sum_{i,j,k=1}^3 \sum_{\alpha,\beta=1}^3 P_{i,A}P_{j,B}P_{k,C}P_{\beta,B}P_{\alpha,A}.\]
Then we must determine an upper bound for
\[ \left\|\sum_{\substack{\{i,j,k\}=\{1,2,3\} \\ \alpha \neq i \text{ or } \beta \neq j}} P_{i,A}P_{j,B}P_{k,C}P_{\beta,B}P_{\alpha,A}\right\|_2+\left\|\sum_{\substack{\{i,j,k\} \subsetneq \{1,2,3\} \\ 1 \leq \alpha,\beta \leq 3}} P_{i,A}P_{j,B}P_{k,C}P_{\beta,B}P_{\alpha,A}\right\|_2.\]
Let $\eta_1$ be the first term and $\eta_2$ be the second term. We will first bound $\eta_1$. If $\{i,j,k\}=\{1,2,3\}$ and $\beta=j$ but $\alpha \neq i$, then either $\alpha=j$, in which case one immediately obtains $\|P_{i,A}P_{j,B}P_{k,C}P_{j,B}P_{\alpha,A}\|_2 \leq \|P_{j,B}P_{j,A}\|_2$, or $\alpha=k$, in which case one can use Proposition \ref{proposition: commutators of projections} with the PVMs $\{P_{c,A}\}_{c=1}^3$ and $\{P_{c,B}\}_{c=1}^3$ to obtain
\begin{align*}
\|P_{i,A}P_{j,B}P_{k,C}P_{j,B}P_{\alpha,A}\|_2 &\leq \|P_{i,A}P_{j,B}P_{k,C}P_{k,A}P_{j,B}\|_2+\|[P_{j,B},P_{\alpha,A}]\|_2 \\
&\leq \|P_{k,C}P_{k,A}\|_2+4\sum_{a=1}^3 \|P_{a,A}P_{a,B}\|_2.
\end{align*}
Each fixed $k$ appears twice in terms in $\eta_1$ indexed by $i,j,k$ with $\{i,j,k\}=\{1,2,3\}$, while each $j$ appears twice in terms indexed by $i,j,k,\alpha,\beta$ with $\{i,j,k\}=\{1,2,3\}$, $\alpha=j$ and $\beta=j$. Finally, the number of terms indexed by $i,j,k,\alpha,\beta$ such that $\{i,j,k\}=\{1,2,3\}$, $\alpha=k$ and $\beta=j$ is $6$. Thus,
\[ \sum_{\substack{\{i,j,k\}=\{1,2,3\} \\ \alpha \neq i, \, \beta=j}} \|P_{i,A}P_{j,B}P_{k,C}P_{j,B}P_{\alpha,A}\|_2 \leq \sum_{a=1}^3 (26\|P_{a,A}P_{a,B}\|_2+2\|P_{a,A}P_{a,C}\|_2).\]
If $\{i,j,k\}=\{1,2,3\}$ and $\beta \neq j$, then either $\beta=k$, in which case $\|P_{i,A}P_{j,B}P_{k,C}P_{\beta,B}P_{\alpha,A}\|_2 \leq \|P_{k,C}P_{k,B}\|_2$, or $\beta \neq k$, in which case by Proposition \ref{proposition: commutators of projections} with the PVMs $\{P_{c,C}\}_{c=1}^3$ and $\{P_{c,B}\}_{c=1}^3$ and using the fact that $P_{j,B}P_{\beta,B}=0$,
\begin{align*}
\|P_{i,A}P_{j,B}P_{k,C}P_{\beta,B}P_{\alpha,A}\|_2 &\leq \|[P_{\beta,B},P_{k,C}]\|_2+\|P_{i,A}P_{j,B}P_{\beta,B}P_{k,C}P_{\alpha,A}\|_2 \\
&=\|[P_{\beta,B},P_{k,C}]\|_2 \\
&\leq 4 \sum_{a=1}^3 \|P_{a,B}P_{a,C}\|_2.
\end{align*}
Each $k$ appears $6$ times in terms indexed by $i,j,k,\alpha,\beta$ where $\{i,j,k\}=\{1,2,3\}$ and $\beta=k$. There are $36$ terms that have index of the form $i,j,k,\alpha,\beta$ where $\{i,j,k\}=\{1,2,3\}$ and $\beta \neq k$. In total, we obtain
\[ \sum_{\substack{\{i,j,k\}=\{1,2,3\} \\ \beta \neq j}} \|P_{i,A}P_{j,B}P_{k,C}P_{j,B}P_{\alpha,A}\|_2 \leq \sum_{a=1}^3 150\|P_{a,B}P_{a,C}\|_2.\]
Adding our estimates together,
\[\eta_1 \leq \sum_{a=1}^3(26 \|P_{a,A}P_{a,B}\|_2+2\|P_{a,A}P_{a,C}\|_2+150 \|P_{a,B}P_{a,C}\|_2).\]

For $\eta_2$, we first bound $\delta:=\sum_{\{i,j,k\} \subsetneq \{1,2,3\}} \|P_{i,A}P_{j,B}P_{k,C}P_{j,B}P_{i,A}\|_2$. By subtracting and adding extra terms of the form $P_{i,A}P_{i,B}P_{i,C}$ for $i=1,2,3$, we obtain

\begin{align*}
\delta &\leq \left(\left\| \sum_{i,j} P_{i,A}P_{i,B}P_{j,C}\right\|_2 +\left\| \sum_{i,j} P_{i,A}P_{j,B}P_{i,C} \right\|_2+\left\| \sum_{i,j} P_{j,A}P_{i,B}P_{i,C} \right\|_2+2\left\| \sum_i P_{i,A}P_{i,B}P_{i,C}\right\|_2\right) \\
&=\left(\left\| \sum_{i=1}^3 P_{i,A}P_{i,B} \right\|_2+\left\| \sum_{i=1}^3 P_{i,A}P_{i,C}\right\|_2+\left\| \sum_{i=1}^3 P_{i,B}P_{i,C} \right\|_2+2\left\| \sum_{i=1}^3 P_{i,A}P_{i,B}P_{i,C} \right\|_2\right) \\
&\leq \left(\sum_{i=1}^3 (3\|P_{i,A}P_{i,B} \|_2+\|P_{i,A}P_{i,C} \|_2+\|P_{i,B}P_{i,C} \|_2)\right). 
\end{align*}

Summing this over all $9$ choices of $\alpha,\beta \in \{1,2,3\}$ and using the fact that, for a fixed $\alpha,\beta$, one has $\left\| \sum_{i,j,k \subsetneq\{1,2,3\}} P_{i,A}P_{j,B}P_{k,C}P_{\beta,B}P_{\alpha,A}\right\|_2 \leq \left\| \sum_{\{i,j,k\} \subsetneq \{1,2,3\}} P_{i,A}P_{j,B}P_{k,C}\right\|_2$, we obtain
\[ \eta_2\leq 9\delta \leq \sum_{i=1}^3 (27\|P_{i,A}P_{i,B}\|_2+9\|P_{i,A}P_{i,C}\|_2+9\|P_{i,B}P_{i,c}\|_2).\]
Adding up $\eta_1$ and $\eta_2$ yields the upper bound
\[ \left\|1-\sum_{\{i,j,k\}=\{1,2,3\}} P_{i,A}P_{j,B}P_{k,C}P_{j,B}P_{i,A}\right\|_2 \leq 159\sum_{i=1}^3(\|P_{i,A}P_{i,B}\|_2+\|P_{i,A}P_{i,C}\|_2+\|P_{i,B}P_{i,C}\|_2),\]
which is the desired bound.
\end{proof}

Proposition \ref{proposition: sum of cut down projections} will be used to show that, for the control triangle $\{A,B,C\}$ of our graph in Section \ref{section: graph definition}, the operators $S_{i,j,k}=P_{i,A}P_{j,B}P_{k,C}P_{j,B}P_{i,A}$ for $\{i,j,k\}=\{1,2,3\}$ are positive contractions that approximately satisfy the relations of being a $6$-output PVM.

For perturbing positive contractions that are almost projections, we require the following lemma. The original lemma \cite[Lemma~3.4]{KPS18} is stated for matrix algebras, but as it only relies on the existence of a faithful (normal) trace and the Borel functional calculus, it also holds in any von Neumann algebra $\mathcal{M}$ with faithful normal tracial state $\tau$.

\begin{lemma}[\cite{KPS18}]
\label{lemma: perturbing one projection}
Let $A$ be a positive contraction in $\cM$. Then the spectral projection $B$ of $A$ onto the interval $\left[\frac{1}{2},1\right]$ belongs to $\mathcal{M}$ and satisfies
\[ \|A-B\|_2 \leq 2\sqrt{2}\|A-A^2\|_2.\]
\end{lemma}

We will also need \cite[Lemma~3.6]{KPS18}, to perturb a set of positive contractions that almost satisfies the relations of being a PVM (in $2$-norm) to a PVM. The reason that this lemma is needed is that we will be dealing with positive contractions whose sum is close to the identity in $2$-norm, but which might overshoot the identity. Since we would like the estimate of the perturbation distance in terms of how well the original operators satisfy the required relations, we include a proof.

\begin{lemma}
\label{lemma: perturbing almost PVM}
Let $m \geq 2$ and let $A_1,...,A_m$ be positive contractions in $\cM$. Then there exist mutually orthogonal projections $B_1,...,B_m$ in $\cM$ such that $B_1+\cdots+B_m=1$ and
\[ \|B_i-A_i\|_2 \leq \begin{cases} 2\sqrt{2}\|A_1-A_1^2\|_2 & i=1 \\ \displaystyle 19\sum_{j=1}^{i-1} (\|B_j-A_j\|_2+\|A_iA_j\|_2)+2\sqrt{2}\|A_i^2-A_i\|_2 & 1<i<m \\ \displaystyle 38\sum_{j=1}^{m-1} (\|B_j-A_j\|_2+\|A_mA_j\|_2)+4\sqrt{2}\|A_m^2-A_m\|_2+\left\| 1- \sum_{i=1}^m A_i\right\|_2 & i=m \end{cases}\]
\end{lemma}

\begin{proof}
If $m=1$, then we use Lemma \ref{lemma: perturbing one projection} and obtain an orthogonal projection $B_1$ such that $\|A_1-B_1\| \leq 2\sqrt{2}\|A_1-A_1^2\|_2$.

Suppose that we have already obtained mutually orthogonal projections $B_1,...,B_k$, where $1 \leq k<m$. We first note that, for each $i=1,...,k$,
\[ \|A_{k+1}B_i\|_2 \leq \|A_{k+1}(B_i-A_i)\|_2+\|A_{k+1}A_i\|_2 \leq \|B_i-A_i\|_2+\|A_{k+1}A_i\|_2.\]
Set $B=1-(B_1+\cdots+B_k)$, which is an orthogonal projection. Then it follows that
\[ \|BA_{k+1}B-A_{k+1}\|_2 \leq \|BA_{k+1}B-BA_{k+1}\|_2+\|BA_{k+1}-A_{k+1}\|_2 \leq 2 \|BA_{k+1}-A_{k+1}\|_2.\]
Simplifying the last expression and using the triangle inequality yields
\[ \|BA_{k+1}B-A_{k+1}\|_2 \leq 2\sum_{i=1}^k \|B_iA_{k+1}\|_2 \leq 2\sum_{i=1}^k (\|B_i-A_i\|_2+\|A_{k+1}A_i\|_2).\]
To apply Lemma \ref{lemma: perturbing one projection}, we show that the positive contraction $BA_{k+1}B$ is almost a projection. As a first step, we observe that
\begin{align*}
\|(BA_{k+1}B)^2-A_{k+1}^2\|_2&=\|BA_{k+1}BBA_{k+1}B-A_{k+1}A_{k+1}\|_2 \\
&\leq \|BA_{k+1}BBA_{k+1}B-A_{k+1}BA_{k+1}B\|_2+\|A_{k+1}BA_{k+1}B-A_{k+1}A_{k+1}\|_2 \\
&\leq 2\|BA_{k+1}B-A_{k+1}\|_2.
\end{align*}
As a result, we obtain the estimate
\begin{align*}
\|(BA_{k+1}B)^2-BA_{k+1}B\|_2 &\leq \|(BA_{k+1}B)^2-A_{k+1}^2\|_2+\|A_{k+1}^2-A_{k+1}\|_2+\|A_{k+1}-BA_{k+1}B\|_2 \\
&\leq 3\|BA_{k+1}B-A_{k+1}\|_2+\|A_{k+1}^2-A_{k+1}\|_2 \\
&\leq 6\sum_{i=1}^k (\|B_i-A_i\|_2+\|A_{k+1}A_i\|_2)+\|A_{k+1}^2-A_{k+1}\|_2.
\end{align*}
By Lemma \ref{lemma: perturbing one projection}, the projection $B_{k+1}=\chi_{\left[\frac{1}{2},1\right]}(BA_{k+1}B)$ belongs to $\cM$ and satisfies
\[ \|B_{k+1}-BA_{k+1}B\|_2 \leq 12\sqrt{2}\sum_{i=1}^k (\|B_i-A_i\|_2+\|A_{k+1}A_i\|_2)+2\sqrt{2}\|A_{k+1}^2-A_{k+1}\|_2.\]
Moreover, since $BA_{k+1}B$ is orthogonal to $B_1,...,B_k$, so is $B_{k+1}$. Then the distance from $B_{k+1}$ to $A_{k+1}$ can be estimated as
\begin{align*}
\|B_{k+1}-A_{k+1}\|_2 &\leq \|B_{k+1}-BA_{k+1}B\|_2+\|BA_{k+1}B-A_{k+1}\|_2 \\
&\leq (12\sqrt{2}+2)\sum_{i=1}^k (\|B_i-A_i\|_2+\|A_{k+1}A_i\|_2)+2\sqrt{2}\|A_{k+1}^2-A_{k+1}\|_2.
\end{align*}
The second case in the lemma statement follows by induction.
For the final case, assume $B_1,...,B_m$ are obtained as above. Then since $B_1,...,B_m$ are mutually orthogonal, $1-\sum_{i=1}^m B_i$ is an orthogonal projection. Thus, if we define $\widetilde{B}_i=B_i$ for all $1 \leq i \leq m-1$ and $\widetilde{B}_m=1-\sum_{i=1}^{m-1} B_i$, then $(\widetilde{B}_1,...,\widetilde{B}_m)$ is a PVM in $\cM$ with $\|\widetilde{B}_i-A_i\|_2=\|B_i-A_i\|_2$ for all $1 \leq i \leq m-1$ and
\[ \|\widetilde{B}_m-A_m\|_2 \leq \left\| 1-\sum_{i=1}^m A_i \right\|_2+2\sum_{i=1}^{m} \|B_i-A_i\|_2.\]
The final case of the lemma statement follows.
\end{proof}

\begin{remark}
Suppose that $\delta>0$ and that $A_1,...,A_m$ are positive contractions in $\cM$ such that $\sum_{i=1}^m \|A_i^2-A_i\|_2<\delta$, $\sum_{1 \leq i<j \leq m} \|A_iA_j\|_2<\delta$ and $\left\|1-\sum_{i=1}^m A_i\right\|_2<\delta$. Then Lemma \ref{lemma: perturbing almost PVM} yields a PVM $B_1,...,B_m$ in $\cM$ such that $\sum_{i=1}^m \|B_i-A_i\|_2<f(m,\delta)$ for some function $f$. Let $f(1,\delta)=2\sqrt{2}\delta$. From Lemma \ref{lemma: perturbing almost PVM}, one obtains for each $1<i<m$ that
\[ \sum_{j=1}^i \|B_i-A_i\|_2 \leq 20\sum_{j=1}^{i-1} \|B_j-A_j\|_2+19\sum_{1 \leq i<k \leq j} \|A_iA_k\|_2+2\sqrt{2}\sum_{j=1}^i \|A_j^2-A_j\|_2.\]
Since $19+2\sqrt{2}<22$, we have
\[ \sum_{j=1}^i \|B_i-A_i\|_2 \leq 20f(i-1,\delta)+22\delta.\]
So we can take $f(i,\delta)=20f(i-1,\delta)+22\delta$. However, one can see that this upper bound is exponential in $i$ (and hence polynomial in $2^i$), so that $f(m,\delta)$ is exponential in $m$. We do not know if we can arrange for $f$ to be a polynomial in $m$, rather than in $2^m$.
\end{remark}

For translating approximate synchronous strategies for a non-local game to approximate synchronous strategies for the associated $3$-coloring game, the following special case of Lemma \ref{lemma: perturbing almost PVM}. The estimate is simpler since we are assuming that the operators involved are already orthogonal projections. We include a proof for convenience.

\begin{lemma}
\label{lemma: perturbing two almost orthogonal projections}
Let $A,B \in \cM$ be orthogonal projections. Then there exists an orthogonal projection $C$ in $\cM$ such that $AC=0$ and $\|B-C\|_2 \leq (8\sqrt{2}+2)\|AB\|_2$.
\end{lemma}

\begin{proof}
Define $D=1-A$. Then by properties of the $2$-norm,
\[ \|DBD-B\|_2 \leq \|DBD-BD\|_2+\|BD-B\|_2\leq 2\|DB-B\|_2=2\|AB\|_2.\]
Since $B^2=B$ and $D^2=D$, it follows that
\begin{align*}
\|(DBD)^2-DBD\|_2 &\leq \|DBDDBD-BDBD\|_2+\|BDBD-B\|_2 \\
&\leq 2\|DBD-B\| \\
&\leq 4\|AB\|_2.
\end{align*}
Define $C=\chi_{\left[\frac{1}{2},1\right]}(DBD)$, which is an orthogonal projection in $\cM$. Then by Lemma \ref{lemma: perturbing one projection}, we have $\|C-DBD\|_2 \leq 2\sqrt{2}(4\|AB\|_2)=8\sqrt{2}\|AB\|_2$. Since $DBD$ is already orthogonal to $A$, so is $C$. Finally, by the triangle inequality we obtain the estimate
\[ \|C-B\|_2 \leq \|C-DBD\|_2+\|DBD-B\|_2 \leq (8\sqrt{2}+2)\|AB\|_2,\]
as desired.
\end{proof}

\section{The Graph Construction}
\label{section: graph definition}

In this section, we recall the graph construction from \cite{Ha24}, as we will use this graph in the next two sections. Throughout, we assume that $\mathcal{G}=(I,O,\lambda)$ is a synchronous game with $|I|=n$ and $|O|=m \geq 3$. Then the graph $G_{\lambda}$ is constructed as follows. First, start with a triangle $\Delta=\{A,B,C\}$. For each $1 \leq \alpha \leq m-2$ and $1 \leq x \leq n$, one defines $R_{\alpha,x}$ to be a copy of $K_3 \times K_3$ with vertices denoted by $v(i,j,\alpha,x)$, such that $v(i,j,\alpha,x) \sim v(k,\ell,\alpha,x)$ if and only if exactly one of $i=k$ or $j=\ell$ holds. We perform the graph gluing of $R_{\alpha,x}$ with $\Delta$ by identifying vertex $v(1,2,\alpha,x)$ with vertex $B$ for all copies $R_{\alpha,x}$ (i.e. for all $1 \leq x \leq n$ and $1 \leq \alpha \leq m-2$). Define
\begin{equation}
\widehat{v}(a,x)=\begin{cases} v(1,1,1,x) & a=1 \\ v(2,1,a-1,x) & 2 \leq a \leq m-1 \\ v(2,2,m-2,x) & a=m, \end{cases}
\end{equation}
and perform the graph gluing of $R_{\alpha,x}$ to $R_{\alpha+1,x}$ by identifying $v(3,2,\alpha,x)=v(1,1,\alpha+1,x)$ for all $1 \leq \alpha \leq m-3$.

To encode the orthogonality gadgets, first define subsets $\mathcal{E}_{\lambda}=\{ (a,b,x,y) \in \lambda^{-1}(\{0\}): a,b \in \{1,k\}\}$ and $\mathcal{F}_{\lambda}=\{ (a,b,x,y) \in \lambda^{-1}(\{0\}): 2 \leq a,b \leq k-1\}$. For each $(a,b,x,y) \in \mathcal{E}_{\lambda} \cup \mathcal{F}_{\lambda}$, define $Q_{a,b,x,y}$ to be a copy of $K_3 \times K_3$ with vertices denoted by $q(i,j,a,b,x,y)$, with the same adjacency relations as other copies of $K_3 \times K_3$. Then perform gluings of $Q_{a,b,x,y}$ with certain copies of $R_{\alpha,x}$ by the identifications $\widehat{v}(a,x)=q(1,1,a,b,x,y)$ and $\widehat{v}(b,y)=q(2,2,a,b,x,y)$, and
\[ q(1,2,a,b,x,y)=\begin{cases} B & (a,b,x,y) \in \mathcal{E}_{\lambda} \\ C & (a,b,x,y) \in \mathcal{F}_{\lambda}.\end{cases}\]
Next, for each $1 \leq \alpha \leq m-2$ and $1 \leq x \leq n$, we let $T_{\alpha,x}$ be a triangular prism with triangles $\{s(1,\alpha,x),s(2,\alpha,x),s(3,\alpha,x)\}$ and $\{t(1,\alpha,x),t(2,\alpha,x),t(3,\alpha,x)\}$, with $s(i,\alpha,x) \sim t(i,\alpha,x)$ for all $i$. We make the identifications $s(j,\alpha,x)=v(1,j,\alpha,x)$ (so in particular $s(2,\alpha,x)=B$) and $t(2,\alpha,x)=A$. Lastly, we add the adjacency relations
\begin{align}
A &\sim v(3,3,\alpha,x) \text{ for all } 1 \leq \alpha \leq m-2, \, 1 \leq x \leq n, \\
C &\sim v(2,1,\alpha,x) \text{ for all } 1 \leq \alpha \leq m-2, \, 1 \leq x \leq n, \label{va adjacent to C} \\
\widehat{v}(a,x) &\sim \widehat{v}(b,y) \text{ for all } (a,b,x,y) \in \lambda^{-1}(\{0\}) \setminus (\mathcal{E}_{\lambda} \cup \mathcal{F}_{\lambda}) \label{non orthogonality gadget orthogonalities} \\
A &\sim q(3,3,a,b,x,y) \text{ for all } (a,b,x,y) \in \mathcal{E}_{\lambda} \cup \mathcal{F}_{\lambda}.
\end{align}

The subgraphs $R_{\alpha,x}$, along with the triangluar prisms $T_{\alpha,x}$, for $1 \leq \alpha \leq m-2$, are meant to enforce the algebraic conditions of having an $m$-output PVM for each $x$, while the subgraphs $Q_{a,b,x,y}$ are meant to enforce orthogonality relations related to losing $4$-tuples. The losing $4$-tuples not in $\cE_{\lambda} \cup \cF_{\lambda}$ are taken care of by adjacency relation (\ref{non orthogonality gadget orthogonalities}).

\begin{remark}
The graph defined here differs from \cite{Ha24} in three minor details. First, the graph in \cite{Ha24} was defined using an asymmetric rule function $\lambda_{\text{asym}}$ with resulting game $*$-algebra being isomorphic to the original. This assumption was only cosmetic and meant to decrease the number of vertices of the graph $G_{\lambda}$; the construction in this paper works in the exact same way when the rule function is not asymmetric.

The second difference is that, in \cite{Ha24}, it was assumed that $\lambda(a,a,x,x)>0$ for all $a,x$. This assumption is not needed, as the approximate coloring assignment in the next section will still work for tuples of the form $(a,a,x,x)$ in $\lambda^{-1}(\{0\})$.

The third difference is that the sets $\cE_{\lambda}$ and $\cF_{\lambda}$ were defined in \cite{Ha24} in such a way to avoid tuples of the form $(a,b,x,x)$ with $a \neq b$. Again, this assumption was not needed, but was there to decrease the number of total vertices slightly. In the case of this paper, including these tuples in $\cE_{\lambda}$ and $\cF_{\lambda}$ simplifies the approximate orthogonality estimates in the final section.
\end{remark}

\section{From approximate strategies for $\cG$ to approximate strategies for $3$-coloring $G_{\lambda}$}
\label{section: from G to Glambda}

In this section, we start with a a synchronous game $\cG=(I,O,\lambda)$ with $|I|=n$ and $|O|=m \geq 3$. Define $G_{\lambda}$ to be the graph associated with $\cG$ in the previous section. We fix a synchronous correlation $p \in C_t^s(n,m)$ of the form $p(a,b|x,y)=\tau(E_{a,x}E_{b,y})$, where $\tau$ is a faithful (normal) tracial state on a von Neumann algebra $\cM$ and, for each $x=1,...,n$, $\{E_{a,x}\}_{a=1}^m$ is a PVM in $\cM$. (Every synchronous correlation arises in this way; see \cite{PSSTW16} for $t \in \{loc,q,qc\}$ and \cite{KPS18} for $t=qa$. In addition, if $t=loc$, we can arrange for $\cM$ to be abelian; if $t=q$, we can arrange for $\cM$ to be finite-dimensional; and if $t=qa$ then we can arrange for $\cM$ to be $\cR^{\cU}$, a tracial ultrapower of the hyperfinite $II_1$ factor $\cR$ using a free ultrafilter $\cU$ over $\bN$.) We assume that the winning probability of this strategy with respect to the uniform distribution $\pi_u(x,y)=\frac{1}{n^2}$ on $I \times I$ satisfies
\[ \omega^s(\cG,p,\pi_u):=\frac{1}{n^2} \sum_{x,y=1}^n \sum_{a,b=1}^m \lambda(a,b|x,y)p(a,b|x,y)=1-\delta.\]
We will show that, with a slight modification, the construction in \cite{Ha24} yields a synchronous strategy in the same model (i.e. loc,q,qa,qc) for $\text{Hom}(G_{\lambda},K_3)$ that almost wins the $3$-coloring game. Moreover, this synchronous strategy can be realized on $\cM$ using the tracial state $\tau$.

To construct this strategy, we construct approximate $3$-colorings of each of the named subgraphs in $G_{\lambda}$, and then compute the winning probability of this strategy with respect to the uniform distribution on vertex pairs that form edges in $G_{\lambda}$. (Note that this means that each edge is counted twice, since the edges of $G_{\lambda}$ are undirected.)

To ease notation, given $s,t \in \{1,...,m\}$, we define $E_{[s,t],x}=\sum_{k=s}^t E_{k,x}$ if $s \leq t$, and $E_{[s,t],x}=0$ if $s>t$. The control triangle $\{A,B,C\}$ is colored $A \mapsto (1,0,0)$, $B \mapsto (0,1,0)$ and $C\mapsto (0,0,1)$. For each $1 \leq x \leq n$ and $1 \leq \alpha \leq m-2$, the subgraph $R_{\alpha,x}$ is colored exactly as in \cite{Ha24}, using the three matrices
\begin{align*}
H_{1,\alpha,x}&=\begin{pmatrix} E_{[1,\alpha],x} & 0 & E_{[\alpha+1,m],x} \\ E_{\alpha+1,x} & E_{[\alpha+2,m],x} & E_{[1,\alpha],x} \\ E_{[\alpha+2,m],x} & E_{[1,\alpha+1],x} & 0 \end{pmatrix} \\
H_{2,\alpha,x}&=\begin{pmatrix} 0 & 1 & 0 \\ 1-E_{\alpha+1,x} & 0 & E_{\alpha+1,x} \\ E_{\alpha+1,x} & 0 & 1-E_{\alpha+1,x} \end{pmatrix} \\
H_{3,\alpha,x}&=\begin{pmatrix} E_{[\alpha+1,m],x} & 0 & E_{[1,\alpha],x} \\ 0 & E_{[1,\alpha+1],x} & E_{[\alpha+2,m],x} \\ E_{[1,\alpha],x} & E_{[\alpha+2,m],x} & E_{\alpha+1,x} \end{pmatrix}
\end{align*}
Here, the $(i,j)$-entry of $H_{c,\alpha,x}$ is the projection corresponding to the $c$-th color for the vertex $v(i,j,\alpha,x)$ in $R_{\alpha,x}$. One notes that each of the three matrices have projections for entries and all rows and columns sum to $1$, while fixing the vertex and summing over the colors also yields $1$. Moreover, for each $c=1,2,3$, the corresponding products of $(i,j)$-entries from $H_{c,\alpha,x}$ that are in the same row or column are $0$ since $\{E_{a,x}\}_{a=1}^m$ is a PVM. Thus, these assignments constitute a perfect $3$-coloring of $R_{\alpha,x}$ in $\cM$. 

For the triangular prism $T_{\alpha,x}$, the first triangle is $\{v(1,1,\alpha,x),B,v(1,3,\alpha,x)\}$. We have already made the assignments $v(1,1,\alpha,x) \mapsto (E_{[1,\alpha],x},0,E_{[\alpha+1,m],x})$, $B \mapsto (0,1,0)$ and $v(1,3,\alpha,x) \mapsto (E_{[\alpha+1,m],x},0,E_{[1,\alpha],x})$. The other triangle in $T_{\alpha,x}$ is colored via
\[ t(1,\alpha,x) \mapsto (0,E_{[\alpha+1,m],x},E_{[1,\alpha],x}), \, A \mapsto (1,0,0), \, t(2,\alpha,x) \mapsto (0,E_{[1,\alpha],x},E_{[\alpha+1,m],x}).\]
As with $R_{\alpha,x}$, these projections constitute a perfect $3$-coloring for $T_{\alpha,x}$, as the only property used thus far is that $\{E_{a,x}\}_{a=1}^m$ is a PVM for each $x$.

All we have left to consider are edges corresponding to the elements of $\lambda^{-1}(\{0\})$. This will entail coloring the subgraph $Q_{a,b,x,y}$ in the case when $(a,b,x,y) \in \cE_{\lambda} \cup \cF_{\lambda}$, and simply checking the edge $(\widehat{v}(a,x),\widehat{v}(b,y))$ in the case when $(a,b,x,y) \in \lambda^{-1}(\{0\}) \setminus (\cE_{\lambda} \cup \cF_{\lambda})$. To this end, there are three cases:

\textbf{Case 1.} Suppose that $(a,b,x,y) \in \lambda^{-1}(\{0\}) \setminus (\cE_{\lambda} \cup \cF_{\lambda})$. The only edge involved with this disallowed $4$-tuple is the edge $(\widehat{v}(a,x),\widehat{v}(b,y))$. These vertices have already been colored, either by the assignments
\[ \widehat{v}(a,x) \mapsto (E_{a,x},1-E_{a,x},0), \, \widehat{v}(b,y) \mapsto (E_{b,y},0,1-E_{b,y})\]
in the case when $a \in \{1,m\}$ and $b \in \{2,...,m-1\}$, or
\[ \widehat{v}(a,x) \mapsto (E_{a,x},0,1-E_{a,x}), \, \widehat{v}(b,y) \mapsto (E_{b,y},1-E_{b,y},0)\]
in the case when $a \in \{2,...,m-1\}$ and $b \in \{1,m\}$. Either way, projections corresponding to color $2$ or $3$ are orthogonal, and the only way players can lose on this input is if they both respond with color $1$, which occurs with probability $\widetilde{p}(1,1|\widehat{v}(a,x),\widehat{v}(b,y))=\tau(E_{a,x}E_{b,y})=p(a,b|x,y)$.

\textbf{Case 2.} Suppose that $(a,b,x,y) \in \cE_{\lambda}$. Then we color the corresponding subgraph $Q_{a,b,x,y}$ as follows. Since $q(1,1,a,b,x,y)=\widehat{v}(a,x)$ and $q(2,2,a,b,x,y)=\widehat{v}(b,y)$ and since $a,b \in \{1,m\}$ by assumption, we already have the assignments
\[ q(1,1,a,b,x,y) \mapsto (E_{a,x},0,1-E_{a,x}), \, q(2,2,a,b,x,y) \mapsto (E_{b,y},0,1-E_{b,y}).\]
Using Lemma \ref{lemma: perturbing two almost orthogonal projections}, choose a projection $G_{a,b,x,y}$ in $\cM$ such that $E_{a,x}G_{a,b,x,y}=0$ and $\|G_{a,b,x,y}-E_{b,y}\|_2 \leq (8\sqrt{2}+2)\|E_{a,x}E_{b,y}\|_2$. For each $c=1,2,3$, define $J_{c,a,b,x,y}$ to be the matrix whose $(i,j)$-entry is the $c$-th color assignment for the vertex $q(i,j,a,b,x,y)$, given by
\begin{align*}
J_{1,a,b,x,y}&=\begin{pmatrix} E_{a,x} & 0 & 1-E_{a,x} \\ 1-E_{a,x}-G_{a,b,x,y} & E_{b,y} & E_{a,x} \\ G_{a,b,x,y} & 1-E_{b,y} & 0 \end{pmatrix} \\
J_{2,a,b,x,y}&=\begin{pmatrix} 0 & 1 & 0 \\ E_{a,x}+G_{a,b,x,y} & 0 & 1-E_{a,x}-G_{a,b,x,y} \\ 1-E_{a,x}-G_{a,b,x,y} & 0 & E_{a,x}+G_{a,b,x,y} \end{pmatrix} \\
J_{3,a,b,x,y}&=\begin{pmatrix} 1-E_{a,x} & 0 & E_{a,x} \\ 0 & 1-E_{b,y} & G_{a,b,x,y} \\ E_{a,x} & E_{b,y} & 1-E_{a,x}-G_{a,b,x,y} \end{pmatrix}.
\end{align*}

Note that $\sum_{c=1}^3 (J_{c,a,b,x,y})_{i,j}=1$ for each $i,j$, and all the entries are projections since $E_{a,x}G_{a,b,x,y}=0$. We now compute the sum of probabilities that players lose on certain edges. Since two distinct vertices in $K_3 \times K_3$ are connected by an edge if and only if they belong to the same row or column, it is enough to consider products of entries in the same row and in the same column.

For the color $1$, the first row has entries summing to $1$, making the entries pairwise orthogonal. The same thing occurs for each column in $J_{1,a,b,x,y}$. In row $2$, the $(2,1)$ and $(2,3)$ entries are orthogonal, and considering the other two entries gives $\|(1-E_{a,x}-G_{a,b,x,y})E_{b,y}\|_2^2+\|E_{a,x}E_{b,y}\|_2^2$. Since $\|E_{b,y}-G_{a,b,x,y}\|_2 \leq (8\sqrt{2}+2)\|E_{a,x}E_{b,y}\|_2$ and since $1-E_{a,x}-G_{a,b,x,y}$ and $G_{a,b,x,y}$ are (orthogonal) projections, one obtains
\begin{align*}
\|(1-E_{a,x}-G_{a,b,x,y})E_{b,y}\|_2 &\leq \|(1-E_{a,x}-G_{a,b,x,y})(G_{a,b,x,y}-E_{b,y})\|_2 \\
&\leq \|G_{a,b,x,y}-E_{b,y}\|_2 \\
&\leq (8\sqrt{2}+2)\|E_{a,x}E_{b,y}\|_2.
\end{align*} 
Thus, the total sum of probabilities on losing entries for color $1$ in row $2$ is at most
\[ \|(1-E_{a,x}-G_{a,b,x,y})\|_2^2+\|E_{a,x}E_{b,y}\|_2^2 \leq ((8\sqrt{2}+2)^2+1)\|E_{a,x}E_{b,y}\|_2^2=((8\sqrt{2}+2)^2+1)p(a,b|x,y).\]
For row $3$, both entries are orthogonal to the $(3,3)$ entry, so the only term involved is $G_{a,b,x,y}(1-E_{b,y})$. Note that
\[ \|G_{a,b,x,y}(1-E_{b,y})\|_2 \leq \|G_{a,b,x,y}(1-G_{a,b,x,y})\|_2+\|G_{a,b,x,y}(G_{a,b,x,y}-E_{b,y})\|_2 \leq (8\sqrt{2}+2)\|E_{a,x}E_{b,y}\|_2,\]
since the first term is zero. Thus, the sum of probabilities of losing on color $1$, given an edge in row $3$ of $Q_{a,b,x,y}$, is at most $(8\sqrt{2}+2)^2\|E_{a,x}E_{b,y}\|_2^2=(8\sqrt{2}+2)^2p(a,b|x,y)$.

For color $2$, all adjacency rules are perfectly satisfied since $J_{2,a,b,x,y}$ is already a quantum permutation.

For color $3$, as with color $1$, orthogonality already holds for row $1$ and all the columns. Rows $2$ and $3$ combined have the same error terms as color $1$ did, and one obtains the sum of probabilities of losing on color $3$, given edge inputs from $Q_{a,b,x,y}$, being at most $((8\sqrt{2}+2)^2+1)p(a,b|x,y)$. Thus, the sum of losing probabilities for $Q_{a,b,x,y}$ is at most $(4(8\sqrt{2}+2)^2+2)p(a,b|x,y) \leq 712 p(a,b|x,y)$.

\textbf{Case 3:} If $(a,b,x,y) \in \cF_{\lambda}$, then the coloring assignments are identical to Case 2, except colors $2$ and $3$ are swapped. This swap occurs since, in these subgraphs, the $(1,2)$ vertex is identified with $C$, not $B$. The exact same error estimates still hold, and hence the sum of losing probabilities for $Q_{a,b,x,y}$ in this case is also at most $712p(a,b|x,y)$.

\begin{remark}
In the above estimates for $Q_{a,b,x,y}$, one can ignore the case when $x=y$ and $a \neq b$, since in that case, $E_{a,x}$ and $E_{b,x}$ are already orthogonal by assumption. In this case, one obtains $G_{a,b,x,x}=E_{b,x}$ and the matrices $J_{c,a,b,x,x}$ above for $c=1,2,3$ constitute a perfect $3$-coloring of $Q_{a,b,x,y}$.

The approximate $3$-coloring construction for $Q_{a,b,x,y}$ still works for a tuple of the form $(a,a,x,x) \in \lambda^{-1}(\{0\})$. Such a tuple will belong to either $\cE_{\lambda}$ or $\cF_{\lambda}$. In this case, the probability of losing on edges corresponding to $Q_{a,a,x,x}$ is sitll at most $712p(a,a|x,x)$, and no modification is required for handling this case separately.
\end{remark}

We note that adjacencies to the control triangle $\{A,B,C\}$ from the subgraphs $Q_{a,b,x,y}$ are already respected by the approximate $3$-coloring that we have constructed--this is mainly due to the simple assignments of colors to $A,B,C$.

We are now ready to state the main result of this section.

\begin{theorem}
Let $\cG=(I,O,\lambda)$ be a synchronous non-local game with $|I|=n$ and $|O|=m \geq 3$. Let $G_{\lambda}$ be the graph associated with $\cG$. Let $t \in \{loc,q,qa,qc\}$. Suppose that $(p(a,b|x,y))=(\tau(E_{a,x}E_{b,y})) \in C_t^s(n,m)$ is a synchronous $t$-correlation, where $\tau$ is a faithful normal tracial state on a von Neumann algebra $\cM$ and $\{E_{a,x}\}_{a=1}^m$ is a PVM in $\cM$ for each $1 \leq x \leq n$. If $\ee \geq 0$ and $\omega_t^s(\cG,\pi_u,p)=1-\ee$, then there is a synchronous $t$-correlation $\widetilde{p} \in C_t^s(|V(G_{\lambda})|,3)$ with realization in $(\cM,\tau)$ such that
\[ \omega_t^s(\text{Hom}(G_{\lambda},K_3),\pi_{\text{edges}},\widetilde{p}) \geq 1-\frac{356n^2}{|E(G_{\lambda})|}\ee.\]
In particular, if $\omega_t^s(\cG,\pi_u) \geq 1-\ee$, then $\omega_t^s(\text{Hom}(G_{\lambda},K_3),\pi_{\text{edges}}) \geq 1-\frac{356n^2}{|E(G_{\lambda})|}\ee$.
\end{theorem}

\begin{proof}
We define $\widetilde{p}$ as the corresponding synchronous correlation obtained from the approximate $3$-coloring of $G_{\lambda}$ in this section. Since the operators used belong to the same von Neumann algebra as the projections $\{E_{a,x}\}$ and since we are using the same trace $\tau$, it follows that $\widetilde{p} \in C_t^s(|V(G_{\lambda})|,3)$. We compute
\begin{align*}
1-\omega_t^s(\text{Hom}(G_{\lambda},K_3),\pi_{\text{edges}},\widetilde{p})&=\frac{1}{2|E(G_{\lambda})|} \sum_{(\mu,\nu) \in E(G_{\lambda})} \sum_{c=1}^3 \widetilde{p}(c,c|\mu,\nu) \\
&\leq \frac{1}{2|E(G_{\lambda})|} \left( \sum_{(a,b,x,y) \in \lambda^{-1}(\{0\}) \setminus (\cE_{\lambda} \cup \cF_{\lambda})} p(a,b|x,y)\right. \\
&\,\,\,\,\,\,\,\left.+\sum_{(a,b,x,y) \in \cE_{\lambda} \cup \cF_{\lambda}} 712 p(a,b|x,y)\right) \\
&\leq \frac{712}{2|E(G_{\lambda})|} \sum_{(a,b,x,y) \in \lambda^{-1}(\{0\})}  p(a,b|x,y)\\
&=\frac{356n^2}{|E(G_{\lambda})|} (1-\omega_t^s(\cG,\pi_u,p)).
\end{align*}
The last claim follows by taking suprema.
\end{proof}

\begin{remark}
\label{remark: number of edges}
One can determine the number of edges in $G_{\lambda}$ in terms of $n=|I|$, $m=|O|$ and $|\lambda^{-1}(\{0\})|$. Each subgraph $R_{\alpha,x}$, for $1 \leq \alpha \leq m-2$ and $1 \leq x \leq n$, contains $18$ edges, plus one edge from $C$ to $v(2,1,\alpha,x)$ and one edge from $A$ to $v(3,3,\alpha,x)$. The triangular prism $T_{\alpha,x}$ contributes $5$ extra edges. The subgraph $Q_{a,b,x,y}$ contributes $18$ edges, plus one edge between $q(3,3,a,b,x,y)$ and $A$. There is one additional edge of the form $(\what{v}(a,x),\what{v}(b,y))$ for each $(a,b,x,y) \in \lambda^{-1}(\{0\}) \setminus (\cE_{\lambda} \cup \cF_{\lambda})$. Thus, the number of edges in $G_{\lambda}$ is
\[ |E(G_{\lambda})|=25n(m-2)+19|\cE_{\lambda}|+19|\cF_{\lambda}|+|\lambda^{-1}(\{0\}) \setminus (\cE_{\lambda} \cup \cF_{\lambda})|.\]
Since the sizes of $\cE_{\lambda}$ and $\cF_{\lambda}$ are often difficult to compute (except, for example, in graph coloring games), and since they are disjoint subsets of $\lambda^{-1}(\{0\})$, one always has the weaker bounds
\[ 25n(m-2)+|\lambda^{-1}(\{0\})| \leq |E(G_{\lambda})| \leq 25n(m-2)+19|\lambda^{-1}(\{0\})|.\]
Since $\cG$ is synchronous, each tuple $(a,b,x,x)$ for $a \neq b$ and $1 \leq x \leq n$ belongs to $\lambda^{-1}(\{0\})$. Thus, one always has the crude bounds $nm(m-1) \leq |\lambda^{-1}(\{0\})| \leq |O \times O \times I \times I|=m^2n^2$, so that $|E(G_{\lambda})|$ is bounded below and above by a polynomial in $n,m$.

Thus, in the preceding theorem, one can replace $1-\frac{712n^2}{|E(G_{\lambda})|}\ee$ with the somewhat worse bound $1-\text{poly}(n,m)\ee$ for a suitable polynomial.
\end{remark}

The following theorem is immediate from the previous remark and theorem.

\begin{theorem}
\label{theorem: forward direction}
For each $t \in \{loc,q,qa,qc\}$, if $\omega_t^s(\cG,\pi_u) \geq 1-\ee$, then $\omega_t^s(\text{Hom}(G_{\lambda},K_3),\pi_{\text{edges}}) \geq 1-\text{poly}(n,m)\ee$.
\end{theorem}

In the case when $t=q$, having $\omega_q^s(\cG,\pi_u) \geq 1-\ee$ does not necessarily mean that there is any synchronous quantum correlation $p$ satisfying $\omega_q^s(\cG,\pi_u,p) \geq 1-\ee$, since the set of synchronous quantum correlations is not closed in general \cite{KPS18}. However, given any $\delta>0$, there exists $p \in C_q^s(n,m)$ such that $\omega_q^s(\cG,\pi_u,p) \geq 1-\ee-\delta$, and hence there is an element $\widetilde{p}$ of $C_q^s(|V(G_{\lambda})|,3)$ such that $\omega_q^s(\text{Hom}(G_{\lambda},K_3),\pi_{\text{edges}},\widetilde{p}) \geq 1-\text{poly}(n,m)(\ee+\delta)$. Then letting $\delta \to 0^+$ yields the desired outcome.

\section{Approximate $3$-colorings of $G_{\lambda}$}
\label{section: from Glambda to G}

In this section, we show that synchronous strategies that approximately win the $3$-coloring game of the graph $G_{\lambda}$ from Section \ref{section: graph definition} yield synchronous strategies that approximately win the game $\cG$ in the same model.

Before we consider the graph $G_{\lambda}$, we make some preliminary observations about $3$-coloring games. Given a graph $G$, we note that, for a $t$-correlation $p(a,b|x,y)$ (not necessarily synchronous) for the game $\text{Hom}(G,K_3)$, the value of the game with that correlation and the uniform distribution on all question pairs that form edges (this set is size $2|E(G)|$) is given by
\[ \omega(\text{Hom}(G,K_3),\pi_{\text{edges}},p)=\frac{1}{2|E(G)|} \sum_{\substack{x,y \in V(G) \\ x \sim y}} \sum_{\substack{a,b \in \{1,2,3\} \\ a \neq b}} \lambda(a,b,x,y) p(a,b|x,y),\]
where
\[ \lambda(a,b,x,y)=\begin{cases} 0 & x=y, \, a \neq b \\ 0 & x \sim y, \, a=b \\ 1 & \text{otherwise}. \end{cases}\]

If $p(a,b|x,y)=\la P_{a,x}Q_{b,y}\psi,\psi \ra$ where $\psi \in \cH$ is a unit vector and $\{P_{a,x}: 1 \leq a \leq 3\}$ and $\{Q_{a,x}: 1 \leq a \leq 3\}$ are PVMs in $\bofh$ with $[P_{a,x},Q_{a,x}]=0$ for all $a,x$, then one can re-define
\[ \widetilde{P}_{a,x}=\bigoplus_{\sigma \in S_3} P_{\sigma(a),x} \text{ and } \widetilde{Q}_{b,y}=\bigoplus_{\sigma \in S_3} Q_{\sigma(b),y}.\]
With $\widetilde{\psi}=\frac{1}{\sqrt{6}} \bigoplus_{\sigma \in S_3} \psi$, and with $\widetilde{p}(a,b|x,y)=\la \widetilde{P}_{a,x}\widetilde{Q}_{b,y}\widetilde{\psi},\widetilde{\psi} \ra$, one has
\[ \widetilde{p}(a,b|x,y)=\frac{1}{6} \sum_{\sigma \in S_3} \la P_{\sigma(a),x}Q_{\sigma(b),y}\psi,\psi \ra=\frac{1}{6}\sum_{\sigma \in S_3} p(\sigma(a),\sigma(b)|x,y).\]
Thus, $\widetilde{p}(a,a|x,y)$ is independent of $a$, and $\widetilde{p}(a,b|x,y)$ for distinct $a,b$ does not depend on the choice of distinct $a,b \in \{1,2,3\}$. Moreover, we have
\begin{align*}
\omega(\text{Hom}(G,K_3),\pi_{\text{edges}},\widetilde{p})&=\frac{1}{2|E(G)|} \sum_{\substack{x,y \in V(G) \\ x \sim y}} \sum_{\substack{a,b \in \{1,2,3\} \\ a \neq b}}\lambda(a,b,x,y)\widetilde{p}(a,b|x,y) \\
&=\frac{1}{2|E(G)|} \sum_{\substack{x,y \in V(G) \\ x \sim y}} \sum_{\substack{a,b \in \{1,2,3\} \\ a \neq b}} \lambda(a,b,x,y)p(a,b|x,y) \\
&=\omega(\text{Hom}(G,K_3),\pi_{\text{edges}},p).
\end{align*}

Thus, when we consider approximately winning (synchronous) strategies for the $3$-coloring game, we may restrict our attention to synchronous approximate strategies of the form $p(a,b|x,y)=\tau(P_{a,\mu},P_{b,\nu})$, for a faithful normal trace $\tau$ on a von Neumann algebra $\cM$ and PVMs $\{P_{a,\mu}\}_{a=1}^3$ for each $\mu \in V(G_{\lambda})$, where $p(a,a|\mu,\nu)$ is independent of $a$, and $p(a,b|\mu,\nu)$ for distinct $a \neq b$ and $(\mu,\nu) \in E(G)$ is independent of the choice of distinct $a,b$. The latter implies that $\|P_{a,\mu}P_{b,\nu}\|_2=\tau(P_{a,\mu}P_{b,\nu}P_{a,\mu})^{\frac{1}{2}}=\tau(P_{a,\mu}P_{b,\nu})^{\frac{1}{2}}$ is independent of the choice of $a,b$, so long as $a \neq b$. Assuming that $\omega(\text{Hom}(G,K_3),\pi_{\text{edges}},p) \geq 1-\ee$, one has
\[ \frac{1}{2|E(G)|}\sum_{\substack{\mu,\nu \in V(G) \\ \mu \sim \nu}} \|P_{a,\mu}P_{a,\nu}\|_2^2=\frac{1}{2|E(G)|}\sum_{\substack{\mu,\nu \in V(G) \\ \mu \sim \nu}} p(a,a|\mu,\nu) \leq \frac{\varepsilon}{3}.\]
Then using convexity and the relation of the $2$-norm to the $1$-norm in Euclidean space, one obtains
\begin{align*}
\left( \frac{1}{2|E(G)|} \sum_{\substack{\mu,\nu \in V(G) \\ \mu \sim \nu}} \|P_{a,\mu}P_{a,\nu}\|_2^{\frac{1}{2}} \right)^2 &\leq \frac{1}{2|E(G)|} \sum_{\substack{\mu,\nu \in V(G) \\ \mu \sim \nu}} \|P_{a,\mu}P_{a,\nu}\|_2 \\
&\leq  \left(\frac{1}{2|E(G)|}\sum_{\substack{\mu,\nu \in V(G) \\ \mu \sim \nu}} \|P_{a,\mu}P_{a,\nu}\|_2^2\right)^{\frac{1}{2}} \\
&\leq \left(\frac{\ee}{3}\right)^{\frac{1}{2}}.
\end{align*}
Taking square roots, multiplying both sides by $|E(G)|$, and replacing $\frac{2}{3^{\frac{1}{4}}}$ with $2$, we obtain:
\begin{equation}
\sum_{\substack{\mu,\nu \in V(G) \\\mu \sim \nu}} \|P_{a,\mu}P_{a,\nu}\|_2^{\frac{1}{2}} \leq 2|E(G)| \ee^{\frac{1}{4}}. \label{main reverse estimate}
\end{equation}
This inequality will be important for us in upcoming estimates.

Now we turn our attention to the graph $G_{\lambda}$. Given a triangle $R$ in $G_{\lambda}$, we will define \[\zeta(R)=\sum_{\substack{\nu,\mu \in V(R) \\ \nu \neq \mu}} \|P_{i,\nu}P_{i,\mu}\|_2,\]
and
\[ \eta(R)=\left\| 1-\sum_{\nu \in V(R)} P_{i,\nu}\right\|_2.\] 
By our assumption on $p$, these quantities do not depend on the choice of $i \in \{1,2,3\}$.

Next, given a triangular prism $T$ in $G_{\lambda}$ formed by the two triangles $\{p,q,r\}$ and $\{s,t,u\}$ with $p \sim s$, $q \sim t$ and $r \sim u$, we will define
\[ \xi(T)=\|P_{i,p}P_{i,s}\|_2+\|P_{i,q}P_{i,t}\|_2+\|P_{i,r}P_{i,u}\|_2,\]
which also does not depend on the choice of $i \in \{1,2,3\}$. Lastly, if $(x,y)$ is an edge in $G_{\lambda}$, then we define
\[ \theta(x,y)=\|P_{i,x}P_{i,y}\|_2,\]
which again does not depend on $i \in \{1,2,3\}$. To simplify notation, for an edge between $(i,j,a,b,x,y)$ and $(k,\ell,a,b,x,y)$ we will let $\theta(q((i,j),(k,\ell),(a,b,x,y)))$ be the quantity $\theta(q(i,j,a,b,x,y),q(k,\ell,a,b,x,y))$.

The quantity $\zeta(R)$ tracks how well the rules are followed for the coloring game in the triangle $R$, and the quantity $\xi(T)$ tracks how well the rules are followed for adjacencies between vertices in the triangular prism $T$ that do not arise from the triangles in $T$.

The following proposition is a restatement of Propositions \ref{proposition: commutators of projections}, \ref{proposition: approx quantum permutation} and \ref{proposition: almost commuting almost quantum permutations} for our synchronous approximate strategy for $\text{Hom}(G_{\lambda},K_3)$.

\begin{proposition}
\label{proposition: restatement for notation} Let $(\mu,\nu)$ be an edge in $G_{\lambda}$; let $S$ be a triangle in $G_{\lambda}$; and let $T$ be a triangular prism in $G_{\lambda}$ consisting of triangles $T_1$ and $T_2$. Then:
\begin{enumerate}
\item For each $i,j \in \{1,2,3\}$ with $i \neq j$, we have
\[ \|[P_{i,\mu},P_{j,\nu}]\|_2 \leq 12 \theta(\mu,\nu).\]
\item For each $i=1,2,3$, we have
\[ \left\| \sum_{w \in V(S)} P_{i,w}-1\right\|_2 \leq 3\zeta(S)^{\frac{1}{2}}.\]
\item
If $i,j \in \{1,2,3\}$ and $v,w \in V(T)$ are non-adjacent vertices, then
\[ \|[P_{i,v},P_{j,w}]\|_2 \leq \begin{cases} 6\zeta(T_1)^{\frac{1}{2}}+6\zeta(T_2)^{\frac{1}{2}}+4\xi(T) & i=j \\ 36\zeta(T_1)^{\frac{1}{2}}+36\zeta(T_2)^{\frac{1}{2}}+24\xi(T) & i \neq j. \end{cases}\]
\end{enumerate}
\end{proposition}

\begin{proof}
For (1), by Proposition \ref{proposition: commutators of projections} with the PVMs $\{ P_{i,\mu}\}_{i=1}^3$ and $\{P_{i,\nu}\}_{i=1}^3$, we have for each pair $i,j$ with $i \neq j$ that
\[ \|[P_{i,\mu},P_{j,\nu}]\|_2 \leq 2\sum_{i=1}^3 \|[P_{i,\mu},P_{i,\nu}]\|_2 \leq 4\sum_{i=1}^3 \|P_{i,\mu}P_{i,\nu}\|_2.\]
Since $\|P_{i,\mu}P_{i,\nu}\|_2=\theta(\mu,\nu)$ for all $i$, (1) follows.

Note that the expression in (2) is exactly $\eta(S)$. By Proposition \ref{proposition: approx quantum permutation} we immediately obtain $\eta(S) \leq 3\zeta(S)^{\frac{1}{2}}$, as desired.

It remains to prove (3). In the case when $i=j$, then the first case of Proposition \ref{proposition: almost commuting almost quantum permutations} yields
\[ \|[P_{i,v},P_{j,w}]\|_2 \leq 2\eta(T_1)+2\eta(T_2)+4\xi(T).\]
Applying (2) then gives the desired inequality. The case when $i \neq j$ is similar and uses the third case of Proposition \ref{proposition: almost commuting almost quantum permutations}.
\end{proof}

We now set some notation regarding the subgraph $R_{\alpha,x}$ of $G_{\lambda}$. Since each $R_{\alpha,x}$ is a copy of $K_3 \times K_3$ for $1 \leq \alpha \leq m-2$ and $1 \leq x \leq n$, we will define $R_{i,(\alpha,x)}$ to be the triangle $\{v(i,j,\alpha,x): j=1,2,3\}$ in $R_{\alpha,x}$, and we will define $R_{(\alpha,x),j}$ to be the triangle $\{v(i,j,\alpha,x): i=1,2,3\}$ in $R_{\alpha,x}$. Similarly, we will define the triangular prism $R_{i,j,(\alpha,x)}$ to be given by the triangles $R_{i,(\alpha,x)}$ and $R_{j,(\alpha,x)}$. We will also define $S_{\alpha,x}$ to be the triangle $\{s(1,\alpha,x),A,s(2,\alpha,x)\}$ in the triangular prism $T_{\alpha,x}$ (the other triangle is $R_{1,(\alpha,x)}$).

For each triangle $T$ in $G_{\lambda}$, one has
\begin{equation}
(\zeta(T))^{\frac{1}{2}}=\left( \sum_{\substack{\mu,\nu \in V(T) \\ \mu \neq \nu}} \|P_{1,\mu}P_{1,\nu}\|_2\right)^{\frac{1}{2}} \leq \sum_{\substack{\mu,\nu \in V(T) \\ \mu \neq \nu}} \|P_{1,\mu}P_{1,\nu}\|_2^{\frac{1}{2}}, \label{zeta estimate triangle}
\end{equation}
using properties of the square root function. Moreover, for any two projections $P_{1,\nu}$ and $P_{1,\mu}$ for $\mu,\nu \in V(G_{\lambda})$, one has $\tau(P_{1,\nu}P_{1,\mu}) \in [0,1]$, so that $\tau(P_{1,\nu}P_{1,\mu})^{\frac{1}{4}}=\|P_{1,\nu}P_{1,\mu}\|_2^{\frac{1}{2}} \geq \tau(P_{1,\nu}P_{1,\mu})$. These facts will be used freely throughout this section.

We start by showing that, for the control triangle $\Delta=\{A,B,C\}$, the operators $S_{i,j,k}=P_{i,A}P_{j,B}P_{k,C}P_{j,B}P_{i,A}$, for $\{i,j,k\}=\{1,2,3\}$, are positive contractions that approximately satisfy the relations of a $6$-output PVM in $2$-norm.

\begin{lemma}
\label{lemma: cut-down contractions are an almost-PVM}
For each choice of $i,j,k$ with $\{i,j,k\}=\{1,2,3\}$, $S_{i,j,k}=P_{i,A}P_{j,B}P_{k,C}P_{j,B}P_{i,A}$ is a positive contraction, and the following estimates hold:
\[ \left\| 1-\sum_{\{i,j,k\}=\{1,2,3\}} S_{i,j,k}\right\|_2 \leq \frac{477}{2}\zeta(\Delta).\]
Moreover,
\[ \sum_{\{i,j,k\}=\{1,2,3\}} \|S_{i,j,k}^2-S_{i,j,k}\|_2 \leq 72\zeta(\Delta),\]
and
\[ \sum_{\substack{\{i_1,j_1,k_1\}=\{i_2,j_2,k_2\}=\{1,2,3\} \\ (i_1,j_1,k_1) \neq (i_2,j_2,k_2)}} \|S_{i_1,j_1,k_1}S_{i_2,j_2,k_2}\|_2 \leq 36\zeta(\Delta).\]
\end{lemma}

\begin{proof}
The fact that each $S_{i,j,k}$ is a positive contraction is immediate. By Proposition \ref{proposition: sum of cut down projections}, 
\begin{align*}
\left\| I-\sum_{\{i,j,k\}=\{1,2,3\}} S_{i,j,k}\right\|_2&\leq 159 \sum_{i=1}^3 (\|P_{i,A}P_{i,B}\|_2+\|P_{i,A}P_{i,C}\|_2+\|P_{i,B}P_{i,C}\|_2) \\
&=\frac{477}{2}\zeta(\Delta).
\end{align*}

For each choice of $\{i,j,k\}=\{1,2,3\}$, define $\mu_{A,B,ij}=\|[P_{i,A},P_{j,B}]\|_2$, $\mu_{B,C,jk}=\|[P_{j,B},P_{k,C}]\|_2$ and $\mu_{A,C,ik}=\|[P_{i,A},P_{k,C}]\|_2$. Then using the fact that each of $P_{i,A}$, $P_{j,B}$ and $P_{k,C}$ are projections,
\begin{align*}
S_{i,j,k}^2=(P_{i,A}P_{j,B}P_{k,C}P_{j,B}P_{i,A})^2 =&P_{i,A}P_{j,B}P_{k,C}P_{j,B}P_{i,A}P_{j,B}P_{k,C}P_{j,B}P_{i,A} \\
\approx_{\mu_{B,C,jk}}& P_{i,A}P_{j,B}P_{k,C}P_{j,B}P_{i,A}P_{j,B}P_{k,C}P_{i,A} \\
\approx_{\mu_{A,B,ij}}& P_{i,A}P_{j,B}P_{k,C}P_{j,B}P_{i,A}P_{k,C}P_{i,A} \\
\approx_{\mu_{A,C,ik}}& P_{i,A}P_{j,B}P_{k,C}P_{j,B}P_{k,C}P_{i,A} \\
\approx_{\mu_{B,C,jk}}& P_{i,A}P_{j,B}P_{k,C}P_{j,B}P_{i,A}=S_{i,j,k}.
\end{align*}

We have (by Proposition \ref{proposition: commutators of projections} with the PVMs $\{P_{1,A},P_{2,A},P_{3,A}\}$ and $\{P_{1,B},P_{2,B},P_{3,C}\}$), that
\[ \sum_{\substack{1 \leq i,j \leq 3 \\ i \neq j}} \mu_{A,B,i,j} \leq 12\sum_{i=1}^3 \|[P_{i,A},P_{i,B}]\|_2 \leq 24\sum_{i=1}^3 \|P_{i,A}P_{i,B}\|_2.\]
Similar statements hold for the $\mu_{B,C,j,k}$'s and $\mu_{A,C,i,k}$'s. Putting everything together, for each choice of $\{i,j,k\}=\{1,2,3\}$, we have $\|S_{i,j,k}^2-S_{i,j,k}\|_2 \leq \mu_{A,B,ij}+\mu_{A,C,ik}+2\mu_{B,C,jk}$, so that
\[ \sum_{\{i,j,k\}=\{1,2,3\}} \|S_{i,j,k}^2-S_{i,j,k}\|_2 \leq 48 \sum_{i=1}^3 (\|P_{i,A}P_{i,B}\|_2+\|P_{i,B}P_{i,C}\|_2+\|P_{i,A}P_{i,C}\|_2)=72\zeta(\Delta).\]

Next, we check the condition on the sums of $2$-norms of products of distinct $S_{i,j,k}$'s. Suppose that $\{i_1,j_1,k_1\}=\{i_2,j_2,k_2\}=\{1,2,3\}$ and that $(i_1,j_1,k_1) \neq (i_2,j_2,k_2)$. Since $i_1,j_1,k_1$ must all be distinct and $i_2,j_2,k_2$ must all be distinct, we either have $i_1 \neq i_2$, or $i_1=i_2$ and $j_1 \neq j_2$.
If $i_1 \neq i_2$, then we have
\[ S_{i_1,j_1,k_1}S_{i_2,j_2,k_2}=(P_{i_1,A}P_{j_1,B}P_{k_1,C}P_{j_1,B}P_{i_1,A})(P_{i_2,A}P_{j_2,B}P_{k_2,C}P_{j_2,B}P_{i_2,A})=0\]
since $P_{i_1,A}P_{i_2,A}=0$, by virtue of the strategy being synchronous and arising from a faithful tracial state. In the case when $i_1=i_2=i$ but $(i_1,j_1,k_1) \neq (i_2,j_2,k_2)$, then since $\{i_1,j_1,k_1\}=\{i_2,j_2,k_2\}=\{1,2,3\}$, we have $j_1=k_2$ and $k_1=j_2$. Then one can write
\begin{align*}
\|S_{i_1,j_1,k_1}S_{i_2,j_2,k_2}\|_2&=\|(P_{i,A}P_{j_1,B}P_{k_1,C}P_{j_1,B}P_{i,A})(P_{i,A}P_{j_2,B}P_{k_2,C}P_{j_2,B}P_{i,A})\|_2 \\
&\leq \|P_{j_1,B}P_{i,A}P_{j_2,B}\|_2 \\
&\leq \|P_{j_1,B}P_{j_2,B}P_{i,A}\|_2+\|[P_{i,A},P_{j_2,B}]\|_2 \\
&=\|[P_{i,A},P_{j_2,B}]\|_2.
\end{align*}
Since there are only $6$ choices of $\{i_1,j_1,k_1\}=\{i_2,j_2,k_2\}=\{1,2,3\}$ with $i_1=i_2$ but $(j_1,k_1) \neq (j_2,k_2)$, one obtains
\[ \sum_{\substack{\{i_1,j_1,k_1\}=\{i_2,j_2,k_2\}=\{1,2,3\} \\ (i_1,j_1,k_1) \neq (i_2,j_2,k_2)}} \|S_{i_1,j_1,k_1}S_{i_2,j_2,k_2}\|_2 \leq \sum_{\substack{1 \leq i,j \leq 3 \\ i \neq j}} \|[P_{i,A},P_{j,B}]\|_2 \leq 24 \sum_{i=1}^3 \|P_{i,A}P_{i,B}\|_2,\]
yielding $\frac{3}{2}(24\zeta(\Delta))=36\zeta(\Delta)$ as an upper bound.
\end{proof}

One can show that each projection corresponding to the vertices $A,B,C$ approximately commute with all other projections in the synchronous strategy given; however, we don't need the full generality of this fact--we need only show that $S_{i,j,k}$ approximately commutes with $P_{i,\what{v}(a,x)}$ for each $\{i,j,k\}=\{1,2,3\}$ and $1 \leq a \leq m$ and $1 \leq x \leq n$. We will split this result into three lemmas, based on how $\widehat{v}(a,x)$ is defined, which is based on the value of $a$.

We first record the following fact, which is immediate when considering the fact that $S_{i,j,k}:=P_{i,A}P_{j,B}P_{k,C}P_{j,B}P_{i,A}$.

\begin{proposition}
\label{proposition: generic commutator with S_{i,j,k}}
For each $X \in \cM$, we have 
\[\|[X,S_{i,j,k}]\|_2 \leq 2\|[X,P_{i,A}]\|_2+2\|[X,P_{j,B}]\|_2+\|[X,P_{k,C}]\|_2.\]
\end{proposition}

We will use Proposition \ref{proposition: generic commutator with S_{i,j,k}} frequently in this section.

\begin{lemma}
For each $\{i,j,k\}=\{1,2,3\}$ and $1 \leq x \leq n$, we have
\[ \|[P_{i,\what{v}(1,x)},S_{i,j,k}]\|_2 \leq 54(\zeta(\Delta))^{\frac{1}{2}}+48(\zeta(R_{1,(1,x)}))^{\frac{1}{2}}+32\xi(T_{1,x})+36\theta(\what{v}(1,x),B).\]
\end{lemma}

\begin{proof}
Note that $\widehat{v}(1,x)=v(1,1,1,x)$. Since $A$ is in a triangular prism with $v(1,1,1,x)$ and these two vertices are non-adjacent, by Proposition \ref{proposition: restatement for notation}, we have
\begin{equation}
\|[P_{i,A},P_{i,v(1,1,1,x)}]\|_2\leq 6 (\zeta(\Delta))^{\frac{1}{2}}+6(\zeta(R_{1,(1,x)}))^{\frac{1}{2}}+4\xi(T_{1,x}). \label{v1 commutator with A}
\end{equation}

Next, note that $v(1,1,1,x)$ and $B$ are adjacent. Using Proposition \ref{proposition: restatement for notation}, we obtain
\begin{equation}
\|[P_{i,\widehat{v}(1,x)},P_{j,B}]\|_2 \leq 12\theta(\widehat{v}(1,x),B). \label{v1 commutator with B}
\end{equation}
Lastly, we have
\begin{align*}
\| [P_{i,\widehat{v}(1,x)},P_{k,C}]\|_2 &\leq 2\left\| 1-\sum_{\nu \in V(\Delta)} P_{k,\nu}\right\|_2+\|[P_{i,\widehat{v}(1,x)},P_{k,A}]\|_2+\|[P_{i,\widehat{v}},P_{k,B}]\|_2 \\
&\leq 6\zeta(\Delta)^{\frac{1}{2}}+\|[P_{i,\widehat{v}(1,x)},P_{k,A}]\|_2+\|[P_{i,\widehat{v}(1,x)},P_{k,B}]\|_2.
\end{align*}
For the second and third terms, we apply Proposition \ref{proposition: restatement for notation} and obtain
\[
\|[P_{i,\widehat{v}(1,x)},P_{k,A}]\|_2 \leq 36\zeta(\Delta)^{\frac{1}{2}}+36\zeta(R_{1,(1,x)})^{\frac{1}{2}}+24\xi(T_{1,x})\]
and
\[\|[P_{i,\widehat{v}(1,x)},P_{k,B}]\|_2 \leq 12\theta(\widehat{v}(1,x),B).\]
Adding these two inequalities gives
\begin{equation}
\|[P_{i,\widehat{v}(1,x)},P_{k,C}]\|_2 \leq 42(\zeta(\Delta))^{\frac{1}{2}}+36(\zeta(R_{1,(1,x)}))^{\frac{1}{2}}+24\xi(T_{1,x})+12\theta(\what{v}(1,x),B). \label{v1 commutator with C}
\end{equation}
The estimate on $\|[P_{i,\what{v}(1,x)},S_{i,j,k}]\|_2$ follows by (\ref{v1 commutator with A})--(\ref{v1 commutator with C}) and an application of Proposition \ref{proposition: generic commutator with S_{i,j,k}}.
\end{proof}

\begin{lemma}
\label{lemma: commutator of va with sijk}
For each $\{i,j,k\}=\{1,2,3\}$, $1 \leq x \leq n$ and $2 \leq a \leq m-1$, we have
\begin{multline*} \|[S_{i,j,k},P_{i,\what{v}(a,x)}]\|_2 \leq 24(\zeta(\Delta))^{\frac{1}{2}}+72(\zeta(R_{1,(a-1),x}))^{\frac{1}{2}}+84(\zeta(R_{2,(a-1,x)}))^{\frac{1}{2}} \\
+56\xi(R_{1,2,(a-1,x)})+16\theta(\what{v}(a,x),C).
\end{multline*}
\end{lemma}

\begin{proof}
Since $2 \leq a \leq m-1$, we have $\widehat{v}(a,x)=v(2,1,a-1,x)$. This vertex is adjacent to $C$ by (\ref{va adjacent to C}), so by Proposition \ref{proposition: restatement for notation}, 
\begin{equation} \|[P_{i,\widehat{v}(a,x)},P_{k,C}]\|_2 \leq 12\theta(\widehat{v}(a,x),C). \label{va commutator with C}
\end{equation}
Next, we note that the triangles $R_{1,(a-1,x)}$ and $R_{2,(a-1,x)}$ form a triangular prism in $G_{\lambda}$, with $\widehat{v}(a,x)$ being a vertex in $R_{2,(a-1,x)}$ and $B$ being a vertex in $R_{1,(a-1,x)}$. Since $\widehat{v}(a,x)$ and $B$ are not adjacent, by Proposition \ref{proposition: restatement for notation} for the triangular prism $R_{1,2,(a-1,x)}$, we have
\begin{equation}
\|[P_{i,\widehat{v}(a,x)},P_{j,B}]\|_2 \leq 36(\zeta(R_{1,(a-1,x)}))^{\frac{1}{2}}+36(\zeta(R_{2,(a-1),x}))^{\frac{1}{2}}+24\xi(R_{1,2,(a-1,x)}). \label{va commutator with B}
\end{equation}

Lastly, we note that $A$ and $\widehat{v}(a,x)$ are not adjacent and do not belong to a triangular prism in $G_{\lambda}$, so we have
\begin{equation}
\|[P_{i,\widehat{v}(a,x)},P_{i,A}]\|_2\leq 2\left\| 1-\sum_{\nu \in V(\Delta)} P_{1,\nu}\right\|_2+\|[P_{i,\widehat{v}(a,x)},P_{i,B}]\|_2+\|[P_{i,\widehat{v}(a,x)},P_{i,C}]\|_2. \label{almost va commutator with A}
\end{equation}
The first term in (\ref{almost va commutator with A}) is exactly $2\eta(\Delta) \leq 6\zeta(\Delta)^{\frac{1}{2}}$ by Proposition \ref{proposition: restatement for notation}. The second term in (\ref{almost va commutator with A}), by Proposition \ref{proposition: restatement for notation} for the triangular prism $R_{1,2,(a-1,x)}$, is
\[ \|[P_{i,\widehat{v}(a,x)},P_{i,B}]\|_2 \leq 6\zeta(R_{2,(a-1,x)})^{\frac{1}{2}}+6\zeta(\Delta)^{\frac{1}{2}}+4\xi(R_{1,2,(a-1,x)}).\]
Since $\widehat{v}(a,x)$ and $C$ are adjacent, the third term in (\ref{almost va commutator with A}) is at most $2\theta(\widehat{v}(a,x),C)$, so we obtain
\begin{equation}
\|[P_{i,\widehat{v}(a,x)},P_{i,A}]\|_2 \leq 12(\zeta(\Delta))^{\frac{1}{2}}+6(\zeta(R_{2,(a-1,x)}))^{\frac{1}{2}}+4\xi(R_{1,2,(a-1,x)})+2\theta(\what{v}(a,x),C). \label{va commutator with A}
\end{equation}

The estimate on $\|[S_{i,j,k},P_{i,\what{v}(a,x)}]\|_2$ follows from Proposition \ref{proposition: generic commutator with S_{i,j,k}}, using inequalities (\ref{va commutator with C}), (\ref{va commutator with B}) and (\ref{va commutator with A}).
\end{proof}

\begin{lemma}
For each $1 \leq x \leq n$ and $\{i,j,k\}=\{1,2,3\}$, 
\begin{multline*}
\|[P_{i,\what{v}(m,x)},S_{i,j,k}]\|_2\leq 66(\zeta(\Delta))^{\frac{1}{2}}+66(\zeta(R_{2,(m-2,x)}))^{\frac{1}{2}}+84(\zeta(R_{1,(m-2,x)}))^{\frac{1}{2}} \\
+32\xi(R_{1,2,(m-2,x)})+4\theta(\widehat{v}(m-1,x),C)+18(\zeta(R_{(m-2,x),3}))^{\frac{1}{2}}\\
+32\xi(T_{m-2,x})+12\theta(v(2,1,m-2,x),C)+2\theta(v(3,3,m-2,x)) \\
+24\theta(\what{v}(m,x),B).
\end{multline*}
\end{lemma}

\begin{proof}
We first note that $\widehat{v}(m,x)=v(2,2,m-2,x)$, so that $\widehat{v}(m,x) \sim B$. By Proposition \ref{proposition: restatement for notation},
\begin{equation} \|[P_{i,\widehat{v}(m,x)},P_{j,B}]\|_2 \leq 12\theta(\widehat{v}(m,x),B). \label{vm commutator with B}
\end{equation}
Next, using the fact that $\eta(R_{2,(m-2,x)})\leq 3 (\zeta(R_{2,(m-2,x)}))^{\frac{1}{2}}$ we have
\begin{equation}
\|[P_{i,\widehat{v}(m,x)},P_{i,A}]\|_2\leq  6(\zeta(R_{2,(m-2,x)}))^{\frac{1}{2}}+\|[P_{i,v(2,1,m-2,x)},P_{i,A}]\|_2+\|[P_{i,v(2,3,m-2,x)},P_{i,A}]\|_2. \label{almost vm commutator with A}
\end{equation}
Since $v(2,1,m-2,x)=\widehat{v}(m-1,x)$, by the proof of Lemma \ref{lemma: commutator of va with sijk}, the second quantity in (\ref{almost vm commutator with A}) is at most $12(\zeta(\Delta))^{\frac{1}{2}}+6(\zeta(R_{2,(a-1,x)}))^{\frac{1}{2}}+4\xi(R_{1,2,(a-1,x)})+2\theta(\widehat{v}(m-1,x),C)$. For the commutator $[P_{i,v(2,3,m-2,x)},P_{i,A}]$, we note that
\[ \|[P_{i,v(2,3,m-2,x)},P_{i,A}]\|_2 \leq 2\eta(R_{(m-2,x),3})+4\theta(v(3,3,m-2,x),A)+\|[P_{i,v(1,3,m-2,x)},P_{i,A}]\|_2.\]
Since $v(1,3,m-2,x)$ and $A$ belong to the triangular prism $T_{m-2,x}$ and are non-adjacent, by Proposition \ref{proposition: restatement for notation} we have
\[ \|[P_{i,v(1,3,m-2,x)},P_{i,A}]\|_2 \leq 6(\zeta(S_{m-2,x}))^{\frac{1}{2}}+6(\zeta(R_{1,(m-2,x)}))^{\frac{1}{2}}+4\xi(T_{m-2,x}).\]
Putting these facts together, we obtain
\begin{multline*}
\|[P_{i,\widehat{v}(m,x)},P_{i,A}]\|_2 \leq 12(\zeta(R_{2,(m-2,x)}))^{\frac{1}{2}}+12(\zeta(\Delta))^{\frac{1}{2}}+6(\zeta(R_{1,(m-2,x)}))^{\frac{1}{2}} \\
+4\xi(R_{1,2,(m-2,x)})+2\theta(\widehat{v}(m-1,x),C) \\
+6(\zeta(R_{(m-2,x),3}))^{\frac{1}{2}}+4\xi(T_{m-2,x}).
\end{multline*}

For the last commutator $[P_{i,\what{v}(m,x)},P_{k,C}]$, we will show that $P_{k,C}$ almost commutes with each of $P_{i,v(2,1,m-2,x)}$ and $P_{i,v(2,3,m-2,x)}$. The first is simply by adjacency--that is to say, $\|[P_{i,v(2,1,m-2,x)},P_{k,C}]\|_2 \leq 12\theta(v(2,1,m-2,x),C)$, by Proposition \ref{proposition: restatement for notation}. For determining a bound on $\|[P_{i,v(2,3,m-2,x)},P_{k,C}]\|_2$, we first determine bounds on commutators when replacing $P_{k,C}$ with $P_{k,A}$ and $P_{k,B}$, respectively. Note that $B$ and $v(2,3,m-2,x)$ are non-adjacent vertices in the triangular prism $R_{1,2,(m-2,x)}$, so by Proposition \ref{proposition: restatement for notation}, one has
\[ \|[P_{i,v(2,3,m-2,x)},P_{k,B}]\|_2 \leq 36(\zeta(R_{1,(m-2,x)}))^{\frac{1}{2}}+36(\zeta(R_{2,(m-2,x)}))^{\frac{1}{2}}+24\xi(R_{1,2,(m-2,x)}).\]
By the exact same estimate, since $A$ and $v(1,3,m-2,x)$ are non-adjacent vertices in the triangular prism $T_{m-2,x}$, the projections $P_{k,A}$ and $P_{i,v(1,3,m-2,x)}$ satisfy
\[ \|[P_{i,v(1,3,m-2,x)},P_{k,A}]\|_2 \leq 36(\zeta(R_{1,(m-2,x)}))^{\frac{1}{2}}+36(\zeta(\Delta))^{\frac{1}{2}}+24\xi(T_{m-2,x}).\]
Using the fact that $v(3,3,m-2,x) \sim A$, it follows that
\begin{multline*}
\|[P_{i,v(2,3,m-2,x)},P_{k,A}]\|_2 \leq 6(\zeta(R_{(m-2,x),3}))^{\frac{1}{2}}+36(\zeta(R_{1,(m-2,x)}))^{\frac{1}{2}}+36(\zeta(\Delta))^{\frac{1}{2}}+24\xi(T_{m-2,x}) \\
+2\theta(v(3,3,m-2,x),A).
\end{multline*}
Therefore, one obtains
\begin{align*}
\|[P_{i,v(2,3,m-2,x)},P_{k,C}]\|_2&\leq 6(\zeta(\Delta)))^{\frac{1}{2}}+\|[P_{i,v(2,3,m-2,x)},P_{k,A}]\|_2+\|[P_{i,v(2,3,m-2,x)},P_{k,B}]\|_2 \\
&\leq 42(\zeta(\Delta))^{\frac{1}{2}}+6(\zeta(R_{(m-2,x),3}))^{\frac{1}{2}}+72(\zeta(R_{1,(m-2,x)}))^{\frac{1}{2}}+24\xi(T_{m-2,x}) \\
&\,\,\,\,\,\,\,+2\theta(v(3,3,m-2,x),A)+36(\zeta(R_{2,(m-2,x)}))^{\frac{1}{2}}+24\xi(R_{1,2,(m-2,x)}).
\end{align*}
Finally, it follows that
\begin{align*}
\|[P_{i,\what{v}(m,x)},P_{k,C}]\|_2&\leq 2\eta(R_{2,(m-2,x)})+\|[P_{i,v(2,1,m-2,x)},P_{k,C}]\|_2+\|[P_{i,v(2,3,m-2,x)},P_{k,C}]\|_2 \\
&\leq 12\theta(v(2,1,m-2,x),C)+42(\zeta(\Delta))^{\frac{1}{2}}+6(\zeta(R_{(m-2,x),3}))^{\frac{1}{2}} \\
&\,\,\,\,\,\,\,+72(\zeta(R_{1,(m-2,x)}))^{\frac{1}{2}}+24\xi(T_{m-2,x})+2\theta(v(3,3,m-2,x)) \\
&\,\,\,\,\,\,\,+42(\zeta(R_{2,(m-2,x)}))^{\frac{1}{2}}+24\xi(R_{1,2,(m-2,x)}).
\end{align*}
The estimate follows by Proposition \ref{proposition: generic commutator with S_{i,j,k}}.
\end{proof}

The following result follows from the previous three lemmas.

\begin{lemma}
\label{lemma: Pi's with compression almost work}
There is a constant $\beta>0$, independent of $n$, $m$ and $|\lambda^{-1}(\{0\})|$, such that the following hold:
\begin{enumerate}
\item For all $\{i,j,k\}=\{1,2,3\}$ and $1 \leq x \leq n$, we have
\[ \sum_{a=1}^m \|[P_{i,\what{v}(a,x)},S_{i,j,k}]\|_2 \leq \beta |E(G_{\lambda})|\ee^{\frac{1}{4}}; \text{ and}\]
\item $\displaystyle\sum_{(a,b,x,y) \in \lambda^{-1}(\{0\})} \|S_{i,j,k}P_{i,\what{v}(a,x)}S_{i,j,k}P_{i,\what{v}(b,y)}S_{i,j,k}\|_2 \leq (\beta |\lambda^{-1}(\{0\})|+\beta) |E(G_{\lambda})|\ee^{\frac{1}{4}}$.
\end{enumerate}
\end{lemma}

\begin{proof}
For the first claim, we note that, for each triangle $T$, triangular prism $R$ and edge $e$ in $G_{\lambda}$,  each term of the form $\zeta(T)$, $\xi(R)$ and $\theta(e)$ consists of sums of the form $\|P_{1,\mu}P_{1,\nu}\|_2$, for certain edges $(\mu,\nu)$ of $G_{\lambda}$. Each edge appears a finite number of times, and one can see that this number does not depend on $n$, $m$ or $\lambda$. Noting that $\|P_{1,\mu}P_{1,\nu}\|_2 \leq \|P_{1,\mu}P_{1,\nu}\|_2^{\frac{1}{2}}$ and using the concavity of the square root function, it readily follows that there is some constant $\beta_1>0$ such that, for all $x$,
\[ \sum_{a=1}^m \|[P_{i,\what{v}(a,x)},S_{i,j,k}]\|_2 \leq \beta_1 \sum_{\substack{\mu,\nu \in V(G) \\ \mu \sim \nu}} \|P_{1,\mu}P_{1,\nu}\|_2^{\frac{1}{2}} \leq 2\beta_1 |E(G_{\lambda})| \ee^{\frac{1}{4}},\]
by (\ref{main reverse estimate}).
For the second claim, to simplify notation we let $S=S_{i,j,k}$. If $(a,b,x,y) \in \cE_{\lambda} \cup \cF_{\lambda}$, then
\begin{align*}
\|SP_{i,\what{v}(a,x)}SP_{i,\what{v}(b,y)}S\|_2&\leq \|[P_{i,\what{v}(a,x)},S]\|_2+\|S^2P_{i,\what{v}(a,x)}P_{i,\what{v}(b,y)}S\|_2 \\
&\leq \|[P_{i,\what{v}(a,x)},S]\|_2+\|SP_{i,\what{v}(a,x)}P_{i,\what{v}(b,y)}S\|_2. \\
\end{align*}
Next, since $\what{v}(a,x)=q(1,1,a,b,x,y)$ and $\what{v}(b,y)=q(2,2,a,b,x,y)$, we can replace $P_{i,\what{v}(b,y)}$ with the other projections from column $2$ of $Q_{a,b,x,y}$ corresponding to color $i$ and obtain
\begin{multline*}
\|SP_{i,\what{v}(a,x)}SP_{i,\what{v}(b,y)}S\|_2\leq \|[P_{i,\what{v}(a,x)},S]\|_2+\eta(Q_{(a,b,x,y),2})+\theta(q((1,1),(1,2),(a,b,x,y))) \\
+\|SP_{i,\what{v}(a,x)}P_{i,q(3,2,a,b,x,y)}S\|_2.
\end{multline*}
Then replacing $q(3,2,a,b,x,y)$ with $q(3,1,a,b,x,y)$ and $q(3,3,a,b,x,y)$ by introducing the term $\eta(Q_{3,(a,b,x,y)}) \leq 3(\zeta(Q_{3,(a,b,x,y)}))^{\frac{1}{2}}$, it follows that
\begin{align*}
\|SP_{i,\what{v}(a,x)}SP_{i,\what{v}(b,y)}S\|_2&\leq \|[P_{i,\what{v}(a,x)},S]\|_2+3(\zeta(Q_{(a,b,x,y),2}))^{\frac{1}{2}}+\theta(q((1,1),(1,2),(a,b,x,y))) \\
&\,\,\,\,\,\,\,+3(\zeta(Q_{3,(a,b,x,y)}))^{\frac{1}{2}}+\theta(q((1,1),(3,1),(a,b,x,y))) \\
&\,\,\,\,\,\,\,+\|SP_{i,q(1,1,a,b,x,y)}P_{i,q(3,3,a,b,x,y)}S\|_2.
\end{align*}
Using the fact that $S=P_{i,A}P_{j,B}P_{k,C}P_{j,B}P_{i,A}$, the last term satisfies
\[ \|SP_{i,q(1,1,a,b,x,y)}P_{i,q(3,3,a,b,x,y)}S\|_2 \leq \|P_{i,q(3,3,a,b,x,y)}P_{i,A}\|_2=\theta(q(3,3,a,b,x,y),A).\]
Thus, for any $(a,b,x,y) \in \cE_{\lambda} \cup \cF_{\lambda}$ we have the estimate
\begin{align*}
\|SP_{i,\what{v}(a,x)}SP_{i,\what{v}(b,y)}S\|_2&\leq \|[P_{i,\what{v}(a,x)},S]\|_2+3(\zeta(Q_{(a,b,x,y),2}))^{\frac{1}{2}}+\theta(q((1,1),(1,2),(a,b,x,y))) \\
&\,\,\,\,\,\,\,+3(\zeta(Q_{3,(a,b,x,y)}))^{\frac{1}{2}}+\theta(q((1,1),(3,1),(a,b,x,y))) \\
&\,\,\,\,\,\,\,+\theta(q(3,3,a,b,x,y),A).
\end{align*}
For the tuples $(a,b,x,y) \in \lambda^{-1}(\{0\}) \setminus (\cE_{\lambda} \cup \cF_{\lambda})$ we simply have
\[ \|SP_{i,\what{v}(a,x)}SP_{i,\what{v}(b,y)}S\|_2 \leq \|[P_{i,\what{v}(a,x)},S]\|_2+\theta(\what{v}(a,x),\what{v}(b,y)).\]
The commutators of the form $\|[P_{i,\what{v}(a,x)},S]\|_2$ are counted once per element of $\lambda^{-1}(\{0\})$, so a similar argument as for the first claim shows that there is a $\beta_2>0$, independent of $n$, $m$ and $\lambda^{-1}(\{0\})$, such that
\[ \sum_{(a,b,x,y) \in \lambda^{-1}(\{0\})} \|SP_{i,\what{v}(a,x)}SP_{i,\what{v}(b,y)}S\|_2 \leq (\beta_2 |\lambda^{-1}(\{0\})|+\beta_2)|E(G_{\lambda})|\ee^{\frac{1}{4}}.\]
The lemma statement follows by letting $\beta=\max\{2\beta_1,\beta_2\}$.
\end{proof}

Our next goal is to show that, for each $1 \leq x \leq n$, the sum 
\[\sum_{a=1}^m \sum_{\{i,j,k\}=\{1,2,3\}} S_{i,j,k}P_{i,\what{v}(a,x)}S_{i,j,k}\]
is close to $1$ in $2$-norm. Then we can perturb and obtain our synchronous strategy for $\cG$.

\begin{lemma}
\label{lemma: almost POVMs with respect to Q}
There is a constant $D>0$, independent of $n$, $m$ and $|\lambda^{-1}(\{0\})|$, such that, for each $1 \leq x \leq n$, we have
\[ \left\| 1-\sum_{\{i,j,k\}=\{1,2,3\}} \sum_{a=1}^m S_{i,j,k}P_{i,\what{v}(a,x)}S_{i,j,k}\right\|_2 \leq D |E(G_{\lambda})|\ee^{\frac{1}{4}}. \]
\end{lemma}

\begin{proof}
We first obtain a bound on $\left\| S_{i,j,k}\left(1-\sum_{a=1}^m P_{i,\what{v}(a,x)}\right)S_{i,j,k}\right\|_2$ for each $\{i,j,k\}=\{1,2,3\}$. For ease of notation, fix $\{i,j,k\}=\{1,2,3\}$ and set $S=S_{i,j,k}$. Suppose that $1 \leq \alpha \leq m-2$. We observe that
\begin{align*}
SP_{i,v(1,1,\alpha,x)}S+SP_{i,v(2,1,\alpha,x)}S&=S(P_{i,v(1,1,\alpha,x)}+P_{i,v(2,1,\alpha,x)})S \\
&\approx_{\eta(R_{(\alpha,x),1})} S(1-P_{i,v(3,1,\alpha,x)})S \\
&\approx_{\eta(R_{3,(\alpha,x)})} S(P_{i,v(3,2,\alpha,x)}+P_{i,v(3,3,\alpha,x)})S.
\end{align*}
Lastly, note that $v(3,3,\alpha,x) \sim A$, so we have $\|P_{1,v(3,3,\alpha,x)}S\|_2 \leq \|P_{1,v(3,3,\alpha,x)}P_{1,A}\|_2=\theta(v(3,3,\alpha,x),A)$. It follows that, for each $1 \leq \alpha \leq m-2$,
\[ \|S(P_{i,v(1,1,\alpha,x)}+P_{i,v(2,1,\alpha,x)}-P_{i,v(3,2,\alpha,x)})S\|_2 \leq \eta(R_{(\alpha,x),1})+\eta(R_{3,(\alpha,x)})+\theta(v(3,3,\alpha,x),A).\]
Using induction, the fact that $\widehat{v}(1,x)=v(1,1,1,x)$ and $\widehat{v}(a,x)=v(2,1,a-1,x)$ for $2 \leq a \leq k-1$, and the fact that $v(3,2,\alpha,x)=v(1,1,\alpha+1,x)$ for all $1 \leq \alpha \leq m-3$, it follows that
\[\left\| \sum_{a=1}^{m-1} SP_{i,\widehat{v}(a,x)}S -SP_{i,v(3,2,m-2,x)}S \right\|_2 \leq \sum_{\alpha=1}^{m-2} (\eta(R_{(\alpha,x),1})+\eta(R_{3,(\alpha,x)})+\theta(v(3,3,\alpha,x),A)).\]
Similarly, since $\what{v}(m,x)=v(2,2,m-2,x)$,
\begin{align*}
\left\| \sum_{a=1}^m SP_{i,\widehat{v}(a,x)}S-S^2\right\|_2&\leq \left\| \sum_{a=1}^{m-1} SP_{i,\widehat{v}(a,x)}S-SP_{i,v(3,2,m-2,x)}S\right\|_2 \\
&\,\,\,\,\,\,+\left\| SP_{i,v(3,2,m-2,x)}S+SP_{i,v(2,2,m-2,x)}S-S^2\right\|_2 \\
&\leq \sum_{\alpha=1}^{m-2} (\eta(R_{(\alpha,x),1})+\eta(R_{3,(\alpha,x)})+\theta(v(3,3,\alpha,x),A)) \\
&\,\,\,\,\,\,\,+\eta(R_{(m-2,x),2}).
\end{align*}
Applying this estimate to each $\{i,j,k\}=\{1,2,3\}$ and previous estimates, we obtain
\begin{align*}
\left\| 1-\sum_{\substack{\{i,j,k\}=\{1,2,3\} \\ 1 \leq a \leq m}}S_{i,j,k}P_{i,\what{v}(a,x)}S_{i,j,k}\right\|_2&\leq \left\|1-\sum_{\{i,j,k\}=\{1,2,3\}} S_{i,j,k}\right\|_2+\sum_{\{i,j,k\}=\{1,2,3\}} \|S_{i,j,k}^2-S_{i,j,k}\|_2 \\
&\,\,\,\,\,\,\,+\sum_{\{i,j,k\}=\{1,2,3\}} \left\|S_{i,j,k}\left(1-\sum_{a=1}^m P_{i,\what{v}(a,x)}\right)S_{i,j,k}\right\|_2 \\
\end{align*}
Combining this estimate with the fact that \[\eta(T) \leq 3(\zeta(T))^{\frac{1}{2}} \leq \sum_{\substack{\mu,\nu \in V(T) \\ \mu \neq \nu}} \|P_{1,\mu}P_{1,\nu}\|_2^{\frac{1}{2}}\] for each triangle $T$ in $G_{\lambda}$, and noting that $\theta(\mu,\nu)=\|P_{1,\mu}P_{1,\nu}\|_2 \leq \|P_{1,\mu}P_{1,\nu}\|_2^{\frac{1}{2}}$ for each $(\mu,\nu) \in E(G_{\lambda})$, one can see that each term involving $P_{1,\mu}P_{1,\nu}$ appears a number of times that is independent of the question set, answer set or rule function of the original game $\cG$. The result follows.
\end{proof}

\begin{lemma}
\label{lemma: A_{a,x} almost projection} Let $1 \leq x \leq n$, and let $1 \leq a \leq m$ and $\{i,j,k\}=\{1,2,3\}$. Then
\[ \sum_{a=1}^m \sum_{\{i,j,k\}=\{1,2,3\}} \|(S_{i,j,k}P_{i,\what{v}(a,x)}S_{i,j,k})^2-S_{i,j,k}P_{i,\what{v}(a,x)}S_{i,j,k}\|_2 \leq (144m+\beta|E(G_{\lambda})|)\ee^{\frac{1}{4}}.\]
\end{lemma}

\begin{proof}
Using the fact that $P_{i,\what{v}(a,x)}^2=P_{i,\what{v}(a,x)}$ and $S_{i,j,k}$ is a contraction, we have
\begin{align*}
(S_{i,j,k}P_{i,\what{v}(a,x)}S_{i,j,k})^2&=S_{i,j,k}P_{i,\what{v}(a,x)}S_{i,j,k}S_{i,j,k}P_{i,\what{v}(a,x)}S_{i,j,k} \\
&\approx S_{i,j,k}P_{i,\what{v}(a,x)}S_{i,j,k}P_{i,\what{v}(a,x)}S_{i,j,k} \\
&\approx S_{i,j,k}P_{i,\what{v}(a,x)}S_{i,j,k}^2 \\
&\approx S_{i,j,k}P_{i,\what{v}(a,x)}S_{i,j,k},
\end{align*}
with a total error of at most $2\|S_{i,j,k}^2-S_{i,j,k}\|_2+\|[P_{i,\what{v}(a,x)},S_{i,j,k}]\|_2$. It follows that
\[ \sum_{a=1}^m \sum_{\{i,j,k\}=\{1,2,3\}} \|(S_{i,j,k}P_{i,\what{v}(a,x)}S_{i,j,k})^2-S_{i,j,k}P_{i,\what{v}(a,x)}S_{i,j,k}\|_2 \leq (144m+\beta|E(G_{\lambda})|)\ee^{\frac{1}{4}},\]
by Lemmas \ref{lemma: cut-down contractions are an almost-PVM} and \ref{lemma: Pi's with compression almost work}.
\end{proof}

Combining these basic facts with the previous lemmas, we show that our almost synchronous ``strategy" involving the operators $\sum_{\{i,j,k\}=\{1,2,3\}}S_{i,j,k}P_{i,\what{v}(a,x)}S_{i,j,k}$ almost satisfies the rules of the game.

\begin{lemma}
\label{lemma: approximate A_{a,x} strategy}
There exist constants $K,L>0$ that are independent of $n,m$ and $\lambda^{-1}(\{0\})$ such that
\[ \sum_{\substack{(a,b,x,y) \in \lambda^{-1}(\{0\}) \\ \{i_1,j_1,k_1\}=\{1,2,3\} \\ \{i_2,j_2,k_2\}=\{1,2,3\}}} \|S_{i_1,j_1,k_1}P_{i_1,\what{v}(a,x)}S_{i_1,j_1,k_1}S_{i_2,j_2,k_2}P_{i_2,\what{v}(b,y)}S_{i_2,j_2,k_2}\|_2 \leq (K|\lambda^{-1}(\{0\})|+L)|E(G_{\lambda})|\ee^{\frac{1}{4}}.\]
\end{lemma}

\begin{proof}
Fixing $(a,b,x,y) \in \lambda^{-1}(\{0\})$ in the above sum, the resulting term is at most
\[ \sum_{\substack{\{i_1,j_1,k_1\}=\{1,2,3\} \\ \{i_2,j_2,k_2\}=\{1,2,3\} \\ (i_1,j_1,k_1) \neq (i_2,j_2,k_2)}} \|S_{i_1,j_1,k_1}S_{i_2,j_2,k_2}\|_2+\sum_{\{i,j,k\}=\{1,2,3\}} \|S_{i,j,k}P_{i,\what{v}(a,x)}S_{i,j,k}^2P_{i,\what{v}(b,y)}S_{i,j,k}\|_2.\]
The first term is less than $36\zeta(\Delta)$ by Lemma \ref{lemma: cut-down contractions are an almost-PVM}, while the second term is bounded above by
\[\sum_{\{i,j,k\}=\{1,2,3\}} \|S_{i,j,k}^2-S_{i,j,k}\|_2+\sum_{\{i,j,k\}=\{1,2,3\}} \|S_{i,j,k}P_{i,\what{v}(a,x)}S_{i,j,k}P_{i,\what{v}(b,y)}S_{i,j,k}\|_2.\]
The first term in this sum is at most $72\zeta(\Delta)$ by Lemma \ref{lemma: cut-down contractions are an almost-PVM}. Thus, summing over all $(a,b,x,y) \in \lambda^{-1}(\{0\})$ and applying part (2) of Lemma \ref{lemma: Pi's with compression almost work} gives 
\[\sum_{\substack{(a,b,x,y) \in \lambda^{-1}(\{0\}) \\ \{i_1,j_1,k_1\}=\{1,2,3\} \\ \{i_2,j_2,k_2\}=\{1,2,3\}}} \|S_{i_1,j_1,k_1}P_{i_1,\what{v}(a,x)}S_{i_1,j_1,k_1}S_{i_2,j_2,k_2}P_{i_2,\what{v}(b,y)}S_{i_2,j_2,k_2}\|_2 \leq (K|\lambda^{-1}(\{0\})|+L)|E(G_{\lambda})|\ee^{\frac{1}{4}},\] for certain constants $K,L>0$.
\end{proof}

The next result uses Lemma \ref{lemma: approximate A_{a,x} strategy} and Lemma \ref{lemma: perturbing almost PVM} to yield a synchronous correlation that wins $\cG$ with probability at least $1-\text{poly}(n,2^m)\ee^{\frac{1}{2}}$.

\begin{theorem}
\label{theorem: from almost-PVM strategy to PVM strategy}
There exist PVMs $\{E_{a,x}\}_{a=1}^m$ in $\cM$ for each $1 \leq x \leq n$, with the property that
\[ \frac{1}{n^2} \sum_{(a,b,x,y) \in \lambda^{-1}(\{0\})} \tau(E_{a,x}E_{b,y}) \leq h(n,m)\ee^{\frac{1}{2}},\]
where $h$ is a polynomial in $n$ and $2^m$.
\end{theorem}

\begin{proof}

There are polynomials $g_1,g_2,g_3$ in two variables so that, for each $1 \leq x \leq n$, the collection $\{ S_{i,j,k}P_{i,\what{v}(a,x)}S_{i,j,k}: 1 \leq a \leq m, \, \{i,j,k\}=\{1,2,3\}\}$ is a collection of $6m$ positive contractions satisfying:
\begin{align}
\left\| 1-\sum_{\{i,j,k\}=\{1,2,3\}} \sum_{a=1}^m S_{i,j,k}P_{i,\what{v}(a,x)}S_{i,j,k}\right\|_2 &\leq g_1(n,m)\ee^{\frac{1}{4}}, \label{almost summing to 1} \\
\sum_{a=1}^m \sum_{\{i,j,k\}=\{1,2,3\}} \|(S_{i,j,k}P_{i,\what{v}(a,x)}S_{i,j,k})^2-S_{i,j,k}P_{i,\what{v}(a,x)}S_{i,j,k}\|_2 &\leq g_2(n,m)\ee^{\frac{1}{4}}, \label{almost projections} \\
\sum_{\substack{1 \leq a,b \leq m \\ \{i_1,j_1,k_1\}=\{i_2,j_2,k_2\}=\{1,2,3\} \\ (a,i_1,j_1,k_1) \neq (b,i_2,j_2,k_3)}} \|S_{i_1,j_1,k_1}P_{i_1,\what{v}(a,x)}S_{i_1,j_1,k_1}S_{i_2,j_2,k_2}P_{i_2,\what{v}(b,x)}S_{i_2,j_2,k_2}\|_2 &\leq g_3(n,m)\ee^{\frac{1}{4}}. \label{almost orthogonal}
\end{align}
Applying Lemma \ref{lemma: perturbing almost PVM} to inequalities (\ref{almost summing to 1})--(\ref{almost orthogonal}) yields a PVM $\{ Q_{a,x,i,j,k}: 1 \leq a \leq m, \, \{i,j,k\}=\{1,2,3\}\}$ in $\cM$ such that
\[ \sum_{a=1}^m \sum_{\{i,j,k\}=\{1,2,3\}} \|Q_{a,x,i,j,k}-S_{i,j,k}P_{i,\what{v}(a,x)}S_{i,j,k}\|_2 \leq g_4(n,m)\ee^{\frac{1}{4}},\]
where $g_4$ is polynomial in $n$ and in $2^{6m}=(2^m)^6$ (hence, polynomial in $2^m$). Define 
\[E_{a,x}=\sum_{\{i,j,k\}=\{1,2,3\}} Q_{a,x,i,j,k},\] which is a projection in $\cM$. Then $\{E_{a,x}\}_{a=1}^m$ is a PVM in $\cM$ for each $1 \leq x \leq n$. For each $(a,b,x,y) \in \lambda^{-1}(\{0\})$, we have
\begin{align*}
\|E_{a,x}E_{b,y}\|_2&\leq \sum_{\substack{\{i_1,j_1,k_1\}=\{1,2,3\} \\ \{i_2,j_2,k_2\}=\{1,2,3\}}} \|Q_{a,x,i_1,j_1,k_1}Q_{b,y,i_2,j_2,k_2}\|_2 \\
&\leq \sum_{\substack{\{i_1,j_1,k_1\}=\{1,2,3\} \\ \{i_2,j_2,k_2\}=\{1,2,3\}}}  \|(Q_{a,x,i_1,j_1,k_1}-S_{i_1,j_1,k_1}P_{i_1,\what{v}(a,x)}S_{i_1,j_1,k_1})Q_{b,y,i_2,j_2,k_2}\|_2 \\
&\,\,\,\,\,\,\,+ \sum_{\substack{\{i_1,j_1,k_1\}=\{1,2,3\} \\ \{i_2,j_2,k_2\}=\{1,2,3\}}} \|S_{i_1,j_1,k_1}P_{i_1,\what{v}(a,x)}S_{i_1,j_1,k_1}(Q_{b,y,i_2,j_2,k_2}-S_{i_2,j_2,k_2}P_{i_2,\what{v}(b,y)}S_{i_2,j_2,k_2})\|_2 \\
&\,\,\,\,\,\,\,+\sum_{\substack{\{i_1,j_1,k_1\}=\{1,2,3\} \\ \{i_2,j_2,k_2\}=\{1,2,3\}}} \|S_{i_1,j_1,k_1}P_{i_1,\what{v}(a,x)}S_{i_1,j_1,k_1}S_{i_2,j_2,k_2}P_{i_2,\what{v}(b,y)}S_{i_2,j_2,k_2}\|_2 \\
&\leq 12\sum_{\{i,j,k\}=\{1,2,3\}} \|Q_{a,x,i,j,k}-S_{i,j,k}P_{i,\what{v}(a,x)}S_{i,j,k}\|_2+g_3(n,m)\ee^{\frac{1}{4}}.
\end{align*}
Summing over all $(a,b,x,y) \in \lambda^{-1}(\{0\})$, one obtains
\begin{align*}
\sum_{(a,b,x,y) \in \lambda^{-1}(\{0\})} \tau(E_{a,x}E_{b,y})&=\sum_{(a,b,x,y) \in \lambda^{-1}(\{0\})} \|E_{a,x}E_{b,y}\|_2^2 \\
&\leq \left( \sum_{(a,b,x,y) \in \lambda^{-1}(\{0\})} \|E_{a,x}E_{b,y}\|_2\right)^2 \\
&\leq \text{poly}(n,2^m)\ee^{\frac{1}{2}}.
\end{align*}
Dividing by $n^2$ yields the inequality in the theorem statement.
\end{proof}

\begin{remark}
If one could arrange to have 
\[\sum_{a=1}^m \sum_{\{i,j,k\}=\{1,2,3\}} S_{i,j,k}P_{i,\what{v}(a,x)}S_{i,j,k} \leq 1\] for each $1 \leq x \leq n$, then since these sums are already close to $1$ in $2$-norm, one could easily re-define the positive operators to yield POVMs satisfying Theorem \ref{theorem: from almost-PVM strategy to PVM strategy}. In other words, one would obtain a non-synchronous strategy that wins $\cG$ with respect to $\pi_u$ with probability at least $1-\text{poly}(n,m)\ee^{\frac{1}{2}}$. As a result, the non-synchronous correlation used would be almost synchronous. One could then apply the techniques of Vidick and Marrakchi-de la Salle \cite{Vi21,MdlS23} to yield a synchronous correlation in the same model that wins the game $\cG$ with probability $1-\text{poly}(n,m)\ee^c$ for some $c>0$. However, in our work, it is not guaranteed that the sums $\sum_{a=1}^m \sum_{\{i,j,k\}=\{1,2,3\}} S_{i,j,k}P_{i,\what{v}(a,x)}S_{i,j,k}$ are below $1$, even though they are close to $1$ in $2$-norm. This demonstrates the need for Lemma \ref{lemma: perturbing almost PVM}, and the resulting exponential dependence on $m$, the number of outputs. It would be interesting to know if this can be improved to polynomial dependence (see Problem \ref{problem: perturb almost POVM to POVM}).
\end{remark}

We arrive at the main theorem of this section.

\begin{theorem}
\label{theorem: value of coloring to value of game}
Let $\ee>0$. There exists a polynomial $h$ of two variables such that, whenever $\mathcal{G}$ is a synchronous non-local game with $n \geq 2$ questions and $m$ answers with associated graph $G_{\lambda}$, and whenever $t \in \{loc,q,qa,qc\}$ and $r \in C_t^s(|V(G_{\lambda})|,3)$ satisfies
\[\omega_t^s(\text{Hom}(G_{\lambda},K_3),\pi_{\text{edges}},r) \geq 1-\ee,\]
then there is an element $p \in C_t^s(n,m)$ such that \[\omega_t^s(\cG,\pi_u,p) \geq 1-h(n,2^m)\ee^{\frac{1}{2}}.\]
Moreover, if $r$ can be realized using PVMs in a von Neumann algebra $\cM$ and a faithful (normal) tracial state $\tau$ on $\cM$, then $p$ can also be realized using PVMs in $\cM$ and the tracial state $\tau$.
\end{theorem}

\begin{proof}
The case when $t=qc$ is immediate from Theorem \ref{theorem: from almost-PVM strategy to PVM strategy}. When $t=qa$, the synchronous correlation $r$ can be realized using PVMs in $\mathcal{R}^{\mathcal{U}}$, a tracial ultrapower of the hyperfinite $II_1$ factor von Neumann algebra $\cR$ by a free ultrafilter $\cU$ on $\bN$. \cite{KPS18} (This ultrapower has a unique faithful tracial state.) By Theorem \ref{theorem: from almost-PVM strategy to PVM strategy}, the resulting correlation $p=(\tau(E_{a,x}E_{b,y}))$, being realizable on $\cR^{\cU}$, belongs to $C_{qa}^s(n,m)$ \cite{KPS18}.

Similarly, if $t=loc$, then $r$ can be realized in an abelian von Neumann algebra $\cM$, and Theorem \ref{theorem: from almost-PVM strategy to PVM strategy} yields a synchronous correlation $p$ arising from $\cM$ (in particular, $p \in C_{loc}^s(n,m)$) \cite{PSSTW16}.

If $t=q$, then $r$ can be realized in some finite-dimensional von Neumann algebra $\cM$, so that $p$ is realized in $\cM$ as well by Theorem \ref{theorem: from almost-PVM strategy to PVM strategy}. It follows that $p \in C_q^s(n,m)$ \cite{PSSTW16}.
\end{proof}

Since Theorem \ref{theorem: value of coloring to value of game} only involves synchronous strategies for $\cG$, we easily obtain the following fact.

\begin{theorem}
\label{theorem: synchronous coloring to non-synchronous strategy}
If $\ee>0$ and $t \in \{loc,q,qa,qc\}$, and if $\omega_t^s(\text{Hom}(G_{\lambda},K_3),\pi_{\text{edges}}) \geq 1-\ee$, then $\omega_t(\cG,\pi_u) \geq 1-h(n,2^m)\ee^{\frac{1}{2}}$.
\end{theorem}

\begin{proof}
This result follows from the fact that the synchronous $t$-value of $\cG$, with respect to any prior distribution, is at most the $t$-value of $\cG$ with the same prior distribution.
\end{proof}

\section{The Max-3-Cut problem}
\label{section: max 3 cut}

In this section we apply Theorem \ref{theorem: value of coloring to value of game} to the Max-3-Cut problem of a graph. First, we note a simple proposition that is essentially from \cite{HMNPR24}, which we will use to relate synchronous strategies for $k$-coloring games to quantities that arise when considering certain non-commutative versions of the Max-$k$-Cut problem.

\begin{proposition}
\label{proposition: relate approximate colorings to max-k-cut}
Let $\cM$ be a von Neumann algebra with faithful (normal) tracial state $\tau$, and let $G=(V,E)$ be a graph. If $\{u_j: j \in V\}$ is a collection of unitaries, each of order $k$, and write $u_j=\sum_{a=0}^{k-1} \omega^a e_{a,j}$, where $\{e_{a,j}\}_{a=0}^{k-1}$ is a PVM in $\cM$ and $\omega=\exp\left(\frac{2\pi i}{k}\right)$. Then
\[ \sum_{(i,j) \in E} \left(1-\frac{1}{k}\sum_{s=0}^{k-1} \tau(u_i^s(u_j^s)^*)\right)=|E|\left(1-\sum_{\substack{i,j \in V \\ (i,j) \in E}} \sum_{a=0}^{k-1} \tau(e_{a,i}e_{b,j})\right).\]
In particular, if $p(a,b|i,j)=\tau(e_{a,i}e_{b,j})$ for each $0 \leq a,b \leq k-1$ and $i,j \in V$, then $p=(p(a,b|i,j)) \in C_{qc}^s(|V|,k)$ and
\[ \sum_{(i,j) \in E} \left(1-\frac{1}{k}\sum_{s=0}^{k-1} \tau(u_i^s(u_j^s)^*)\right)=|E| \omega_{qc}^s(\text{Hom}(G,K_k),\pi_{\text{edges}},p).\]
\end{proposition}

\begin{proof}
We first note that
\[ u_i^s=\sum_{a=0}^{k-1} \omega^{as} e_{a,i},\]
since $\{e_{a,i}\}_{a=1}^k$ is a PVM. Thus,
\[\sum_{s=0}^{k-1} \tau(u_i^s(u_j^s)^*)=\sum_{s=0}^{k-1} \sum_{a,b=0}^{k-1} \omega^{s(a-b)} \tau(e_{a,i}e_{b,j}).\]
Since $\omega$ is a primitive $k$-th root of unity, $\sum_{s=0}^{k-1} \omega^{s(a-b)}$ is $0$ if $a \neq b$, and is $k$ if $a=b$. Hence,
\[ \sum_{s=0}^{k-1} \tau(u_i^s(u_j^s)^*)=k\sum_{a=0}^{k-1} \tau(e_{a,i}e_{a,j}).\]
Dividing by $k$ and noting that the sum over all $i,j \in V$ for which $(i,j) \in E$ counts each edge exactly twice,
\begin{align*}
\sum_{(i,j) \in E} \frac{1}{k} \sum_{s=0}^{k-1} \tau(u_i^s(u_j^s)^*)&=\frac{1}{2} \sum_{\substack{i,j \in V \\ (i,j) \in E}} \sum_{a=0}^{k-1}\tau(e_{a,i}e_{b,j}) \\
&=|E| \left( \frac{1}{2|E|} \sum_{\substack{i,j \in V \\ (i,j) \in E}} \sum_{a=0}^{k-1} \tau(e_{a,i}e_{b,j}) \right) \\
&=|E| (1-\omega_{qc}^s(\text{Hom}(G,K_k),\pi_{\text{edges}},p))
\end{align*}
Thus,
\[ \sum_{(i,j) \in E} \left(1-\frac{1}{k}\sum_{s=0}^{k-1} \tau(u_i^s(u_j^s)^*)\right)=|E|\omega_{qc}^s(\text{Hom}(G,K_k),\pi_{\text{edges}},p)\]
\end{proof}

For a graph $G=(V,E)$, the Max-$k$-Cut problem is the problem of partitioning $V$ into $k$ pairwise disjoint subsets $S_1,...,S_k$, for which the number of edges $e=(i,j)$ for which $i$ and $j$ do not belong to the same subset $S_{\ell}$ is maximized. One can readily see that the maximal number is
\[ \text{Max-}k\text{-Cut}(G)=\max \sum_{(i,j) \in E} \left( 1-\frac{1}{k} \sum_{s=0}^{k-1} u_i^s(u_j^s)^*\right),\]
where the maximum is taken over all choices of $u_j \in \{1,\omega,\omega^2,...,\omega^{k-1}\}$ where $\omega=\exp\left(\frac{2\pi i}{j}\right)$. (See, for example, \cite{CMS23}.) Similarly (as done in \cite{HMNPR24}) one can define non-commutative versions.

\begin{definition}
The \textbf{non-commutative Max-$k$-Cut} of a graph $G$ is defined as
\[ \text{NC-Max-}k\text{-Cut}(G)=\sup \sum_{(i,j) \in E} \left(1-\frac{1}{k}\sum_{s=0}^{k-1} \tr_d(u_i^s(u_j^s)^*)\right),\]
where the supremum is taken over all $d \in \mathbb{N}$ and over all elements $\{u_j: j \in V\}$ in $M_d$ satisfying $u_j^*u_j=u_j^k=I_d$ for all $j \in V$.

The \textbf{qc Max-$k$-Cut} of $G$ is defined as
\[ \sup \sum_{(i,j) \in E} \left(1-\frac{1}{k}\sum_{s=0}^{k-1} \tau(u_i^s(u_j^s)^*)\right)\]
where the supremum is taken over all tracial von Neumann algebras $(\cM,\tau)$ and elements $\{u_j: j \in V\}$ in $\cM$ satisfying $u_j^*u_j=u_j^k=1$ for all $j \in V$.
\end{definition}

A few notes are in order. The sum in the supremum is over all edges of $G$, counting each once. The terms involved, while being defined in terms of the vertices of each edge, does not depend on the ordering of the vertices for each edge. We also note that the supremum for the NC Max-k-Cut may not be attained, since the set of synchronous quantum correlations is not closed in general \cite{DPP19}. However, if one were to define the $qa$ Max-k-Cut of $G$ by taking the supremum of the above expression over all order $k$ unitaries in $\cR^{\cU}$, then this supremum would be attained and would be equal to the NC Max-k-Cut.

It is easy to see that the following hold:
\begin{itemize}
\item $\text{Max-}k\text{-Cut}(G)=|E|\omega_{loc}^s(\text{Hom}(G,K_k),\pi_{edges})$;
\item $\text{NC-Max-}k\text{-Cut}(G)=|E|\omega_q^s(\text{Hom}(G,K_k),\pi_{edges})$; and
\item
$qc\text{-Max-}k\text{-Cut}(G)=|E|\omega_{qc}^s(\text{Hom}(G,K_k),\pi_{edges})$.
\end{itemize}

The following uses Theorem \ref{theorem: synchronous coloring to non-synchronous strategy} to address what are known as ``gapped promise problems". For our purposes, the \textbf{gapped $(1,\alpha)$-promise problem} for a synchronous game $\cG$ is the decision problem of determining whether the value of the game $\cG$ (in a particular model $t \in \{loc,q,qa,qc\}$) is equal to $1$, or less than $\alpha$, if promised that one of those two outcomes occur. Such promise problems were a key tool in the proof that $C_{qa}^s(n,m) \neq C_{qc}^s(n,m)$ for large enough $n,m$ \cite{JNVWY20}; indeed, Ji et. al proved that, for each $0<\ee<1$, the gapped $(1,1-\ee)$-promise problem for the quantum value of a synchronous game $\cG$ is undecidable (in fact, RE-hard) \cite{JNVWY20}.

\begin{theorem}
There exists an $\alpha \in (0,1)$ such that, for each $t \in \{q,qa,qc\}$, the gapped $(1,\alpha)$-promise problem for the synchronous $t$-value of the $3$-coloring game is undecidable.
\end{theorem}

\begin{proof}
We first note that, since the synchronous $t$-value of any synchronous game is a supremum over elements of $C_t^s(n,m)$, and since $\overline{C_q^s(n,m)}=C_{qa}^s(n,m)$ \cite{KPS18}, we may assume without loss of generality that $t \in \{qa,qc\}$. Let $0<\ee<1$, and let $\cG$ be a synchronous non-local game, that is promised to either have $t$-value $1$ or $t$-value less than $1-\ee$. We show the gapped $(1,1-\ee)$-promise problem for the $t$-value of a synchronous game $\cG$ can be reduced to the $(1,\alpha)$-promise problem for the synchronous $t$-value of $\text{Hom}(G_{\lambda},K_3)$ with respect to $\pi_{edges}$. We note that $\omega^s_t(\text{Hom}(G_{\lambda},K_3),\pi_{\text{edges}})$ is either $1$ or less than $1-\left(\frac{\ee}{h(n,2^m)}\right)^2$. Indeed, if we had $1-\left(\frac{\ee}{h(n,2^m)}\right)^2<\omega_t^s(\text{Hom}(G_{\lambda},K_3),\pi_{\text{edges}})<1$, then an application of Theorem \ref{theorem: synchronous coloring to non-synchronous strategy} shows that $1-\ee<\omega_t(\cG,\pi_u)<1$ (since a value of $1$ for $\cG$ would yield a perfect $t$-strategy for $\text{Hom}(G_{\lambda},K_3)$ by \cite{Ha24}, since $t \in \{qa,qc\}$). So, to decide the gapped $(1,1-\ee)$-promise problem for $\omega_t(\cG,\pi_u)$, we decide the gapped $(1,\alpha)$-promise problem for $\omega_t^s(\text{Hom}(G_{\lambda},K_3),\pi_{\text{edges}})$, where $\alpha=1-\left(\frac{\ee}{h(n,2^m)}\right)^2$. If the algorithm outputs ``$1$", then $\omega_t(\cG,\pi_u)=1$ (\cite{Ha24}). Otherwise, we obtain $\omega_t^s(\text{Hom}(G_{\lambda},K_3),\pi_{\text{edges}})) \leq \alpha$, forcing $\omega_t(\cG,\pi_u) \leq 1-\ee$. That is to say, the gapped $(1,1-\ee)$-promise problem for the $t$-value of a synchronous non-local game is reduced to the corresponding gapped $(1,\alpha)$-promise problem for the synchronous $t$-value of the $3$-coloring game. The former is undecidable by \cite{JNVWY20}, so the latter must be undecidable as well.
\end{proof}

We now restate this result in terms of the Max-3-Cut problems.

\begin{theorem}
\label{theorem: sharp computability gap}
There exists an $\alpha \in (0,1)$ such that approximating the non-commutative Max-3-Cut (respectively, $qc$ Max-3-Cut) of a graph $G$ within a factor of $\alpha$ is uncomputable, in the following sense: the promise problem of determining whether $\text{NC-Max-}3\text{-Cut}(G)$ (respectively, $qc\text{-Max-}3\text{-Cut}(G)$) is equal to $|E|$ or less than $\alpha|E|$, if promised that one of these occurs, is undedicable.
\end{theorem}

This theorem gives evidence of a ``sharp computability gap". It is known \cite{CMS23} that one can efficiently approximate the non-commutative Max-3-cut of a graph up to a factor of $0.864$, but after this factor the corresponding problem is conjectured to be RE-hard. Our work shows that there is some cut-off where this change in complexity occurs--it does not yield an explicit lower bound for the constant factor (aside from $\alpha>0.864$), but proves the existence of such a factor.

There is also the ``oracularizable" version of the non-commutative max-3-cut (see \cite{MNY22}), where in the supremum one requires that the unitaries $u_j$ and $u_k$ commute if $(j,k)$ is an edge in $G$. This property mirrors what occurs in perfect strategies for $3$-coloring games, and also takes into account the general requirement in constraint systems that variable assignments belonging to a common constraint must commute. (This approach is the same idea as looking at approximate representations of the ``locally commuting" game algebra of the $3$-coloring game for $G$ \cite{HMPS19}.) The results of Theorem \ref{theorem: sharp computability gap} are also true for the oracularizable versions of the non-commutative max-3-cut. Indeed, this follows directly from the fact that the non-commutative max-3-cut is at least as large as the oracularizable version.

Interestingly, our techniques do not say anything about the non-commutative Max-$k$-Cut problem where $k \geq 4$. This problem would be solved if one constructed an equivalence of synchronous non-local games to $k$-coloring games that also preserved approximate winning strategies. However, a lot of the techniques used to construct equivalences of synchronous games to $3$-coloring games \cite{Ji13,Ha24} involve special properties of $3$-coloring game algebras that do not extend to $k$-coloring games. For example, in the game $*$-algebra of the $3$-coloring game of a graph $G$, if $v,w$ are adjacent vertices, then the projections $e_{c,v}$ and $e_{d,w}$ commute with each other for all $c,d=1,2,3$. In triangular prisms, projections corresponding to non-adjacent vertices in a $3$-coloring game $*$-algebra must commute. These special properties need not be satisfied even in $4$-coloring games \cite{MR16} and game $*$-algebras don't offer as much structure in this case, since the $4$-coloring game $*$-algebra for any graph of at least $4$ vertices is non-trivial \cite{HMPS19}. One could instead construct a hereditary $*$-equivalence between $\cG$ and an associated $k$-coloring game--there are equivalences of this form that involve independence and clique numbers (see \cite{Ha24}). However, to the author's knowledge, no explicit hereditary $*$-equivalence of $\cG$ to a $k$-coloring game is known. Hence, we close with the following problem:

\begin{problem}
\label{problem: hereditary equivalence to k-coloring}
Let $k \in \bN$ with $k \geq 4$. Determine whether there is an explicit graph $G_{k,\lambda}$ associated to the synchronous game $\cG=(I,O,\lambda)$ such that $\cG$ and $\text{Hom}(G_{k,\lambda},K_k)$ are hereditarily $*$-equivalent in a way that preserves approximately winning (synchronous) strategies.
\end{problem}

Solving Problem \ref{problem: hereditary equivalence to k-coloring} would solve the following problem:

\begin{problem}
\label{problem: max k cut for k at least 4}
Let $k \in \bN$ with $k \geq 4$. Is there an $\alpha_k \in (0,1)$ for which determining whether the non-commutative (respectively $qc$) Max-$k$-Cut of a graph is $|E|$ or less than $\alpha_k |E|$ undecidable?
\end{problem}

Another question is regarding the exponential dependence on $m$. Solving this issue amounts to the following problem:

\begin{problem}
\label{problem: perturb almost POVM to POVM}
Let $\cM$ be a von Neumann algebra and let $\tau$ be a faithful (normal) tracial state on $\cM$. Let $\delta>0$, and suppose that $A_1,...,A_m \in \cM$ are positive contractions such that $\|1-\sum_{i=1}^m A_i\|_2<\delta$. Is there a POVM $\{B_1,...,B_m\}$ in $\cM$ such that $\|B_i-A_i\|_2 \leq p(m)\delta$ for all $i$, where $p$ is a polynomial?
\end{problem}

If the answer to Problem \ref{problem: perturb almost POVM to POVM} is ``yes", then one could perturb the almost-PVMs $\{S_{i,j,k}P_{i,\what{v}(a,x)}S_{i,j,k}: 1 \leq a \leq m, \, \{i,j,k\}=\{1,2,3\}\}$ to POVMs (in $6m$ outputs) that are almost PVMs, and obtain $m$-output POVMs $\{B_{a,x}\}_{a=1}^m$ for which the resulting tuple $(\tau(B_{a,x}B_{b,y}))_{a,b,x,y}$ would be a correlation in the correct model, that is almost synchronous. (This correlation would be tracial in the sense of \cite{TT24}, but not necessarily synchronous.) Applying the methods of \cite{Vi21} for $t=q$ and \cite{MdlS23} for the other cases would yield a synchronous correlation $p$ for which $\omega_t^s(\cG,\pi_u,p) \geq 1-\text{poly}(n,m)\ee^c$ for some $c>0$ that is independent of $n,m,\ee$. The case $t=qa$ would follow from taking pointwise limits of synchronous correlations in $q$ (see \cite{KPS18,Vi21}).

\section*{Acknowledgements}

S.J. Harris was supported in part by a Lawrence Perko research award from Northern Arizona University. The author would like to thank Hamoon Mousavi and Vern Paulsen for valuable and stimulating discussions.

\end{document}